\documentclass[11pt,a4paper]{article}
\usepackage{amssymb,color}
\usepackage{amsfonts}
\usepackage{amsmath}
\usepackage{amsthm}
\usepackage{graphicx}
\usepackage[dvips]{epsfig}
\usepackage{longtable}
\usepackage{graphics}
\usepackage{color}
\usepackage{algorithm}
\usepackage{algorithmic}
\usepackage{graphicx}				
\usepackage{amssymb}

\newtheorem{theorem}{Theorem}
\newtheorem{lemma}{Lemma}
\newtheorem{corollary}{Corollary}
\newtheorem{definition}{Definition}

\newtheorem{remark}{Remark}

\newcommand{\trans}{{\mathrm{T}}}

\newcommand{\E}{\mathbb{E}}

\newcommand{\VaR}{\operatorname{VaR}}

\DeclareMathOperator*{\essinf}{ess\,inf}

\renewcommand{\P}{\mathbb{P}}
\newcommand{\N}{\mathbb{N}}

\newcommand{\R}{\mathbb{R}}

\newcommand{\Var}{\operatorname{Var}}

\renewcommand{\phi}{\varphi}

\newcommand{\calD}{\mathcal{D}}
\newcommand{\calF}{\mathcal{F}}

\newcommand{\calH}{\mathcal{H}}

\newcommand{\cald}{\mathrm{d}}
\newcommand{\calI}{\mathbb{I}}
\newcommand{\calB}{\mathcal{B}}

\newcommand{\inprob}{\stackrel{{\small\P}}{\to}}

\title{Least Squares Monte Carlo applied to Dynamic Monetary Utility Functions}
\author{Hampus Engsner\footnote{Corresponding author.  Department of Mathematics, Stockholm University, SE-10691 Stockholm, Sweden.  e-mail: hampus.engsner@math.su.se} }

\begin{document}
\maketitle

\begin{abstract}
In this paper we explore ways of numerically computing recursive dynamic monetary risk measures and utility functions. Computationally, this problem suffers from the curse of dimensionality and nested simulations are unfeasible if there are more than two time steps. The approach considered in this paper is to use a Least Squares Monte Carlo (LSM) algorithm to tackle this problem, a method which has been primarily considered for valuing American derivatives, or more general stopping time problems, as these also give rise to backward recursions with corresponding challenges in terms of numerical computation. We give some overarching consistency results for the LSM algorithm in a general setting as well as explore numerically its performance for recursive Cost-of-Capital valuation, a special case of a dynamic monetary utility function. 
\end{abstract}

{\bf keywords:} Monte Carlo algorithms, least-squares regression, multi-period valuation, dynamic utility funcitons

\section{Introduction}

Dynamic monetary risk measures and utility functions, as described for instance in \cite{Artzner-Delbaen-Eber-Heath-Ku-07} and \cite{Cheridito-Delbaen-Kupper-06}, are time consistent if and only if they satisfy a recursive relationship (see for instance \cite{Cheridito-Kupper-11}, \cite{Pelsser-G-16}). In the case of time-consistent valuations of cash flows, often in an insurance setting, (e.g. \cite{Engsner-Lindholm-Lindskog-17}, \cite{Engsner-Lindskog-20}, \cite{Moehr-11}, \cite{Pelsser-G-16}, \cite{Pelsser-Stadje-14}, \cite{Pelsser-G-20}, \cite{Dhaene-17}),  analogous recursions also appear. Recursive relationships also occur as properties of solutions to optimal stopping problems, of which valuation of American derivatives is a special case. It is well known that numerical solutions to these kinds of recursions suffer from "the curse of dimensionality": As the underlying stochastic process generating the flow of information grows high dimensional, direct computations of solutions of these recursions prove unfeasible. 

In order to make objective of this paper more clear, consider a probability space $(\Omega, \calF,\P)$, a $d$-dimensional Markov chain $(S_t)_{t=0}^T$ in $L^2(\calF)$ and its natural filtration $(\calF_t)_{t=0}^T$. We are interested in computing $V_0$ given as the solution to the following recursion
\begin{align}
V_t = \phi_{t}(f(S_{t+1})+V_{t+1}), \quad V_T = 0, \label{eq:recursion_mot}
\end{align}
where, for each $t$, $\phi_t : L^2(\calF_{t+1}) \to L^2(\calF_t)$ is a given law-invariant mapping (see Section \ref{sec:Background} and definition \ref{def:law_invariance}). Recursions such as \eqref{eq:recursion_mot} arise when describing time-consistent dynamic monetary risk measures/utility functions (see e.g. \cite{Cheridito-Delbaen-Kupper-06}).  
Alternatively, we may be interested in computing $V_0$ given as the solution to the following recursion
\begin{align}
V_t = f(S_{t}) \lor \E[V_{t+1} \mid \calF_t], \quad  V_T = f(S_T), \label{eq:optistopping_mot}
\end{align}
where $f:\R^d \to \R$ is a given function. Recursions such as \eqref{eq:optistopping_mot} arise when solving discrete-time optimal stopping problems or valuing American-style derivatives (see e.g. \cite{Egloff-05} and \cite{Longstaff-Schwartz-01}).  
In this article we will focus the recursive expression \eqref{eq:recursion_mot}. In either case, due to the Markovian assumptions, we expect $V_t$ to be determined by some deterministic function of the state $S_t$ at time $t$. The curse of dimensionality can now be succinctly put as the statement that as the dimension $d$ grows, direct computation of $V_t$ often becomes unfeasible. Additionally, brute-force valuation via a nested Monte Carlo simulation, discussed in \cite{Brodie-15} and \cite{Barrera-18}, is only a feasible option when $T=2$, as the number of required simulations would grow exponentially with $T$. One approach to tackle this problem is the Least Squares Monte Carlo (LSM) algorithm, notably used in \cite{Longstaff-Schwartz-01} to value American-style derivatives, and consists of approximating $V_{t+1}$ in either \eqref{eq:recursion_mot} or \eqref{eq:optistopping_mot} as a linear combination basis functions of the state $S_{t+1}$ via least-squares regression. While most often considered for optimal stopping problems (\cite{Longstaff-Schwartz-01}, \cite{Tsitsiklis-vanRoy-01}, \cite{Clement-etal-02}, \cite{Stentoft-04}, \cite{Egloff-05}, \cite{Zanger-09}, \cite{Zanger-13}, \cite{Zanger-18}), it has also been used recently in \cite{Pelsser-G-20} for the purpose of actuarial valuation, with respect to a recursive relationship in line with \eqref{eq:recursion_mot}.

The paper is organized as follows. In Section \ref{sec:Background} we introduce the mathematical definitions and notation that will allow us to describe the LSM algorithm in our setting mathematically, as well as to formulate theoretical results. Section \ref{sec:theoretical_results} contains consistency results with respect to computing \eqref{eq:recursion_mot} both in a general setting, only requiring an assumption of continuity in $L^2$ norm, and for the special case of a Cost-of-Capital valuation, studied in \cite{Engsner-Lindholm-Lindskog-17} and \cite{Engsner-Lindskog-20}, under the assumption that capital requirements are given by the risk measure Value-at-Risk, in line with solvency II. The lack of convenient continuity properties of Value-at-Risk pose certain challenges that are handled.
Section \ref{sec:numerical_results} investigates the numerical performance of the LSM algorithm on valuation problems for a set of models for liability cash flows. Here some effort is also put into evaluating and validating the LSM algorithm's performance, as this is not trivial for the considered cases.

\section{Mathematical setup}\label{sec:Background}

Consider a probability space $(\Omega, \calF,\P)$. On this space we consider two filtrations $(\calF_t)_{t=0}^T$, with $\calF_0=\{\emptyset,\Omega\}$, and $(\calH_t)_{t=0}^T$. The latter filtration is an initial expansion of the former: take $\calD\subset \calF$ and set $\calH_t:=\calF_t\vee \calD$. 
$\calD$ will later correspond to the $\sigma$-field generated by initially simulated data needed for numerical approximations. 
Define $L^2(\calH_t)$ as the spaces of $\calH_t$ measurable random variables $Z$ with $\E[Z^2]<\infty$. The subspace $L^2(\calF_t)\subset L^2(\calH_t)$ is defined analogously. 
All equalities and inequalities between random variables should be interpreted in the $\P$-almost sure sense. 

We assume that the probability space supports a Markov chain $S =(S_t)_{t=0}^T$ on $(\R^d)^T$, where $S_0$ is constant, and an iid sequence $D=(S^{(i)})_{i\in\N}$, independent of $S$, where, for each $i$, $S^{(i)}=(S^{(i)}_t)_{t=0}^T$ has independent components with $\mathcal{L}(S^{(i)}_t)=\mathcal{L}(S_t)$ (equal in distribution). $D$ will represent possible initially simulated data and we set $\calD=\sigma(D)$. The actual simulated data will be a finite sample and we write $D_n := (S^{(i)})_{i=1}^n$. For $Z\in L^2(\calF)$ we write $\Vert Z\Vert_{2} := \E[|Z|^2\mid \calH_0]^\frac{1}{2}$.
Notice that $\Vert Z\Vert_{2}$ is a nonrandom number if $Z$ is independent of $D$. 

The mappings $\rho_t$ and $\phi_t$ appearing in Definitions \ref{def:dynrisk} and \ref{def:dynutil} below can be defined analogously as mappings from $L^p(\calH_{t+1})$ to $L^p(\calH_t)$ for $p\neq 2$. However, $p=2$ will be the relevant choice for the applications treated subsequently. 

\begin{definition}\label{def:dynrisk}
A dynamic monetary risk measure is a sequence $(\rho_t)_{t=0}^{T-1}$ of mappings $\rho_t:L^2(\calH_{t+1})\to L^2(\calH_t)$ satisfying
\begin{align}
& \textrm{if } \lambda\in L^2(\calH_t) \textrm{ and } Y\in L^2(\calH_{t+1}), \textrm{ then } 
\rho_t(Y+\lambda)=\rho_t(Y)-\lambda,  \label{eq:ti_r}\\
& \textrm{if } Y,\widetilde{Y}\in L^2(\calH_{t+1}) \textrm{ and } Y\leq \widetilde{Y}, \textrm{ then } 
\rho_t(Y)\geq \rho_t(\widetilde{Y}),  \label{eq:mo_r}\\
& \rho_t(0)=0. \label{eq:ph_r}
\end{align}
\end{definition}
The elements $\rho_t$ of the dynamic monetary risk measure $(\rho_t)_{t=0}^{T-1}$ are called conditional monetary risk measures.  
  
\begin{definition}\label{def:dynutil}
A dynamic monetary utility function is a sequence $(\phi_t)_{t=0}^{T-1}$ of mappings $\phi_t:L^2(\calH_{t+1})\to L^2(\calH_t)$ satisfying
\begin{align}
& \textrm{if } \lambda\in L^2(\calH_t) \textrm{ and } Y\in L^2(\calH_{t+1}), \textrm{ then } 
\phi_t(Y+\lambda)=\phi_t(Y)+\lambda, \label{eq:ti_u}\\
& \textrm{if } Y,\widetilde{Y}\in L^2(\calH_{t+1}) \textrm{ and } Y\leq \widetilde{Y}, \textrm{ then } 
\phi_t(Y)\leq \phi_t(\widetilde{Y}), \label{eq:mo_u}\\
& \phi_t(0)=0. \label{eq:ph_u}
\end{align}
\end{definition}
Note that if $(\rho_t)_{t=0}^{T-1}$ is a dynamic monetary risk measure, $(\rho_t(- \cdot))_{t=0}^{T-1}$ is a dynamic monetary utility function.  In what follows we will focus on dynamic monetary utility function of the form 
\begin{align}\label{eq:valuation_mapping_def}
\phi_t(Y)=\rho_t(-Y)-\frac{1}{1+\eta_t}\E\big[\big(\rho_t(-Y)-Y\big)^+\mid\calH_t\big],
\end{align}
where $(\rho_t)_{t=0}^T$ is a dynamic monetary risk measures in the sense of Definition \ref{def:dynrisk} and $(\eta_t)_{t=0}^{T-1}$ is a sequence of nonrandom numbers in $(0,1)$. We may consider a more general version of this dynamic monetary utility function by allowing $(\eta_t)_{t=0}^{T-1}$ to be an $(\calF_t)_{t=0}^{T-1}$ adapted sequence, however we choose the simpler version here. That $(\phi_t)_{t=0}^T$ is indeed a dynamic monetary utility function is shown in \cite{Engsner-Lindholm-Lindskog-17}. 

We will later consider conditional monetary risk measures that are conditionally law invariant in the sense of Definition \ref{def:law_invariance} below. Conditional law invariance will then be inherited by $\phi_t$ in \eqref{eq:valuation_mapping_def}. 

\begin{definition}\label{def:law_invariance}
A mapping $\phi_t: L^2(\calH_{t+1}) \to L^2(\calH_{t})$ is called law invariant if $\phi_t(X)=\phi_t(Y)$ whenever $\P(X \in \cdot \mid \calH_t)=\P(Y \in \cdot \mid \calH_t)$.
\end{definition}

We now define the value process corresponding to a dynamic monetary utility function $(\phi_t)_{t=0}^{T-1}$ in the sense of Definition \ref{def:dynutil} with respect to the filtration $(\calF_t)_{t=0}^{T-1}$ instead of $(\calH_t)_{t=0}^{T-1}$. The use of the smaller filtration is due to that a value process of the sort appearing in Definition \ref{def:phi_value} is the theoretical object that we aim to approximate well by the methods considered in this paper.  

\begin{definition}\label{def:phi_value}
Let $L:=(L_t)_{t=1}^T$ with $L_t\in L^2(\calF_t)$ for all $t$.
Let $(\phi_t)_{t=0}^{T-1}$ be a dynamic monetary utility function with respect to $(\calF_t)_{t=0}^{T-1}$. 
Let 
\begin{align}
V_t(L,\phi) := \phi_t(L_{t+1}+V_{t+1}(L,\phi)), \quad V_T(L,\phi) :=0. \label{eq:phivalue_def}
\end{align}
We refer to $V_t(L,\phi)$ as the time $t$ $\phi$-value of $L$.
\end{definition}
Whenever it will cause no confusion, we will suppress the argument $\phi$ in $V_{t}(L,\phi)$ in order to make the expressions less notationally heavy.

\begin{remark}
Letting $\rho$ be a dynamic monetary risk measure and letting $\phi$ be a dynamic monetary utility function, $-\sum_{s=1}^t L_t -V_t(-L,-\rho)$ will be a conditional monetary risk measure on the cash flow $L$ in the sense of \cite{Cheridito-Delbaen-Kupper-06} and likewise $\sum_{s=1}^t L_t +V_t(L,\phi)$ will be a conditional monetary utility function in the sense of \cite{Cheridito-Delbaen-Kupper-06}, with $L$ being interpreted as a process of incremental cash flows. If $L$ is a liability cash flow, we may write the risk measure of $L$ as $\sum_{s=1}^t L_t + V_t(L,\rho)$. Importantly, any time-consistent dynamic monetary utility function/risk measure may be written in this way (see e.g. \cite{Cheridito-Kupper-11}). 
Often convexity or subadditivity is added to the list of desired properties in definitions \ref{def:dynrisk} and \ref{def:dynutil} (see e.g. \cite{Artzner-Delbaen-Eber-Heath-Ku-07}, \cite{Cheridito-Delbaen-Kupper-06}, \cite{Cheridito-Kupper-11}). 
\end{remark}

\subsection{The approximation framework}\label{sec:approix_framework}

For $t=1, \dots, T$, consider a sequence of functions $\{1,\Phi_{t,1},\Phi_{t,2}, \dots \}$, where for each $i\in \N$, $\Phi_{t,i} : \R^d \to \R$ has the property $\Phi_{t,i}(S_t) \in L^p(\calF_t)$ and the set $\{1,\Phi_{t,1}(S_t),\Phi_{t,2}(S_t), \dots \}$ make up a.s. linearly independent random variables. We define the approximation space $\mathcal{B}_{t,N}$ and its corresponding $L^2$ projection operator $P_{\mathcal{B}_{t,N}}:L^2(\calH_t)\to \mathcal{B}_{t,N}$ as follows: for $N\in \N$ and $t\in\{0,\dots,T\}$, 
\begin{align}
&\mathcal{B}_{t,N} := \mathrm{span} \{1, \Phi_{t,1}(S_t), \dots, \Phi_{t,N}(S_t)\}, \label{eq:B_set_def}\\
&P_{\mathcal{B}_{t,N}}Z_{t} := \arg \inf_{B \in \mathcal{B}_{t,N}}\Vert Z_{t}-B\Vert_{2}. \label{eq:projection_def}
\end{align}
Defining $\mathbf{\Phi}_{t,N} := (1, \Phi_{t,1}(S_t), \dots, \Phi_{t,N}(S_t))^\trans$, note that the unique minimizer in \eqref{eq:projection_def} is given by $P_{\mathcal{B}_{t,N}}Z_{t} := \beta_{t,N,Z_t}^\trans\mathbf{\Phi}_{t,N} $, with 
\begin{align}\label{eq:opt_beta_def}
 \beta_{t,N,Z_t} =\E\big[\mathbf{\Phi}_{t,N} \mathbf{\Phi}_{t,N}^{\trans} \mid \calH_0\big]^{-1}\E\big[\mathbf{\Phi}_{t,N} Z_t \mid \calH_0\big],
\end{align}
where the expected value of a vector or matrix is interpreted elementwise. Note that if $Z_t$ in \eqref{eq:opt_beta_def} is independent of the initial data $D$, then $\beta_{t,N,Z_t}$ is a nonrandom vector. Indeed, we will only apply the operator $P_{\mathcal{B}_{t,N}}$ to random variables $Z_t$ independent of $D$.

For each $t$, consider a nonrandom function $z_t$ such that $Z_t = z_t(D_M,S_t)\in L^2(\calH_t)$. For $M\in\N$, let 
\begin{align*}
\mathbf{\Phi}^{(M)}_{t,N}:=\left(\begin{array}{cccc}
1 & \Phi_{t,1}(S^{(1)}_t) & \dots & \Phi_{t,N}(S^{(1)}_t)\\
\vdots & \dots & \dots & \vdots \\
1 & \Phi_{t,1}(S^{(M)}_t) & \dots & \Phi_{t,N}(S^{(M)}_t)
\end{array}\right), 
\end{align*}
\begin{align*}
Z_t^{(M)}:=\left(\begin{array}{c}
z_t(D_M,S^{(1)}_t) \\
\vdots \\
z_t(D_M,S^{(M)}_t)
\end{array}\right)
\end{align*}
and define
\begin{align}
\widehat{\beta}^{(M)}_{t,N,Z_t} &:= \Big(\big(\mathbf{\Phi}^{(M)}_{t,N}\big)^\trans \mathbf{\Phi}^{(M)}_{t,N} \Big)^{-1}\big(\mathbf{\Phi}^{(M)}_{t,N}\big)^\trans Z_t^{(M)}, \label{eq:beta_estimate_def} \\
P_{\mathcal{B}_{t,N}}^{(M)}Z_{t} &:= (\widehat{\beta}^{(M)}_{t,N,Z_t})^\trans \mathbf{\Phi}_{t,N}. \label{eq:M_sample_projection_def}
\end{align}
Notice that $\widehat{\beta}^{(M)}_{t,N,Z_t}$ is independent of $S$ and is the standard OLS estimator of $\beta_{t,N,Z_t}$ in \eqref{eq:opt_beta_def}. Notice also that $\mathbf{\Phi}_{t,N} $ is independent of $D$. 
With the above definitions we can define the Least Squares Monte Carlo (LSM) algorithm for approximating the value $V_0(L, \phi)$ given by \eqref{def:phi_value}. 

Let $(\phi_t)_{t=0}^{T-1}$ be a sequence of law-invariant mappings $\phi_t: L^2(\calH_{t+1}) \to L^2(\calH_t)$. Consider a stochastic process $L:=(L_t)_{t=1}^T$, where $L_t=g_t(S_t)\in L^2(\calF_t)$ for all $t$ for some nonrandom functions $g_t : \R^d \to \R$. The goal is to estimate the values $(V_t(L))_{t=0}^T$ given by Definition \ref{def:phi_value}. Note that the sought values $V_t(L)$ are independent of $D$ and thus, by the law-invariance property and the Markov property, $V_t(L)$ is a function of $S_t$ for each $t$. Now we may describe the LSM algorithm with respect to $N$ basis functions and simulation sample size $M$. The LSM algorithm corresponds to the following recursion:
\begin{align}
\widehat{V}^{(M)}_{N,t}(L) :=P^{(M)}_{\mathcal{B}_{t,N}}\phi_{t}\big(L_{t+1}+\widehat{V}^{(M)}_{N,t+1}(L)\big), 
\quad \widehat{V}^{(M)}_{N,T}(L) :=0. \label{eq:LSM_definition}
\end{align}
Notice that $\widehat{V}^{(M)}_{N,t}$ is a function of the random variables $S_t$ and $S_{u}^{(i)}$ for $1 \leq i \leq M$ and $t+1\leq u\leq T$. In particular, $\widehat{V}^{(M)}_{N,t}\in L^2(\calH_t)$.
In the section below, we will investigate when, and in what manner, $\widehat{V}^{(M)}_{N,t}\in L^2(\calH_t)$ may converge to $V_t\in L^2(\calF_t)\subset L^2(\calH_t)$. For this purpose, we make the additional useful definition:
\begin{align}
\widehat{V}_{N,t}(L) :=P_{\mathcal{B}_{t,N}}\phi_{t}\big(L_{t+1}+\widehat{V}_{N,t+1}(L)\big), 
\quad \widehat{V}_{N,T}(L) :=0. \label{eq:optimal_LSM_definition}
\end{align}
$\widehat{V}_{N,t}(L)$ is to be interpreted as an idealized LSM estimate, where we make a least-squares optimal estimate in each iteration. Note that this quantity is independent of $D$

\section{Consistency results}\label{sec:theoretical_results}

In the following section we prove what essentially are two consistency results for the LSM estimator $\widehat{V}^{(M)}_{N,t}(L)$ along with providing conditions for these to hold. These consistency results are  analogous to Theorems 3.1 and 3.2 in \cite{Clement-etal-02}. The first and simplest result, Lemma \ref{lem:Lp_convergence}, is that if we have a flexible enough class of basis functions, $\widehat{V}_{N,t}(L)$ will asymptotically approach the true value $V_t(L)$. The second consistency result, Theorem  \ref{thm:simulation_convergence}, is that when $N$ is kept fixed, then $\widehat{V}^{(M)}_{N,t}(L)$ will approach the least-squares optimal $\widehat{V}_{N,t}(L)$ for each $t$ as $M$ grows to infinity. Hence, we show that the LSM estimator for a fixed number of basis functions is consistent in the sense that the simulation-based projection operator $P^{(M)}_{\mathcal{B}_{t,N}}$ will approach $P_{\mathcal{B}_{t,N}}$ even in the presence of errors in a multiperiod setting. Lemma \ref{lem:noname2} and Theorem \ref{thm:noname3} furthermore extends these results to the case of a Cost-of-Capital valuation, studied in \cite{Engsner-Lindholm-Lindskog-17} and \cite{Engsner-Lindskog-20}, which here is dependent on the non-continuous risk measure Value-at-Risk. Note from Section \ref{sec:Background} that these results presume simulated data not to be path dependent, in contrast to the results in \cite{Zanger-18}.

We should note that these results do not give a rate of convergence, which is done in the optimal stopping setting in for instance \cite{Egloff-05} and \cite{Zanger-09}. Especially, these papers provide a joint convergence rate in which $M$ and $N$ simultaneously go to infinity, something which is not done here. There are three main reasons for this. First of all, in this paper the purpose is to investigate LSM methods given by standard OLS regression, i.e. we do not want to involve a truncation operator as we believe this would not likely be implemented in practice.  
The use of truncation operators is necessary for the results in \cite{Egloff-05} and \cite{Zanger-09}, although one can handle the case of unbounded cash flows by letting the bound based on the truncation operator suitably go to infinity 
along with $N$ and $M$. Secondly, we believe that the bounds involved in the rates of convergence would be quite large in our case if we applied repeatedly the procedure in \cite{Egloff-05} or \cite{Zanger-09} (see remark  \ref{rem:rate_of_conv}). Thirdly, we want to consider mappings which are $L^2$-continuous (Definition \ref{def:L2_cont}) but not necessarily Lipschitz. In this case it is not clear how convergence can be established other than at some unspecified rate.

\subsection{General convergence results}

We first define a useful mode of continuity that we will require to show our first results on the convergence of the LSM algorithm.
\begin{definition}\label{def:L2_cont}
The mapping $\phi_t: L^2(\calH_{t+1}) \to L^2(\calH_{t})$ is said to be $L^2$-continuous if 
$\Vert X-X_n\Vert_ {2} \inprob 0 \text{ implies } \Vert \phi_t(X) - \phi_t(X_n)\Vert_{2} \inprob 0$.
\end{definition}
Notice that if $(X_n)_{n=1}^\infty$ and $X$ are independent of $D$, the convergence in probability may be replaced by convergence of real numbers. 

We are now ready to formulate our first result on the convergence of the LSM algorithm. The first result essentially says that if we make the best possible estimation in each recursion step, using $N$ basis functions, then, for each $t$, the estimator of $V_t$ will converge in $L^2$ to $V_t$ as $N \to \infty$. This result is not affected by the initial data $D$, as it does not require any simulation-based approximation.

\begin{lemma} \label{lem:Lp_convergence}
For $t=0,\dots,T-1$, let the mappings $\phi_t: L^2(\calH_{t+1}) \to L^2(\calH_{t})$ be $L^2$-continuous and law invariant. 
For $t=1,\dots,T$, let $\bigcup_{n\in \N} \calB_{t,n}$ be dense in the set $\{h(S_t) \mid h:\R^d \to \R,  h(S_t) \in L^2(\calF_t)\}$.
Then, for $t=0,\dots,T-1$, 
\begin{align*}
\Vert \widehat{V}_{N,t}(L) -V_t(L)\Vert_2 \in \R \text{ and } \lim_{N\to\infty}\Vert \widehat{V}_{N,t}(L) -V_t(L)\Vert_2=0.
\end{align*}
\end{lemma}

The second result uses the independence assumptions of $D$ to prove a somewhat technical result for when $P_{\mathcal{B}_{t,N}}^{(M)}$ given by \eqref{eq:M_sample_projection_def} asymptotically approaches the projection $P_{\mathcal{B}_{t,N}}$ given by \eqref{eq:projection_def}.

\begin{lemma}\label{lem:simulation_stability_property}
Let $Z_t = z_t(S_t) \in L^2(\calF_t)$. For each $M \in \N$, let $Z_{M,t} = z_{M,t}(D_M,S_t) \in L^2(\calH_t)$, where $z_{M,t}(D_M, \cdot)$ only depends on $D_M$ through $\{S_{u}^{(i)}: 1 \leq i \leq M, u\geq t+1\}$, i.e. 
\begin{align}
\mathcal{L}\big(z_{M,t}(D_M, S_t^{(i)})\big)=\mathcal{L}\big(z_{M,t}(D_M, S_t)\big).\label{eq:independence_condition}
\end{align}
Then, $\Vert Z_t-Z_{M,t}\Vert_{2} \inprob 0$ as $M\to\infty$ implies 
\begin{align*}
\Vert P_{\mathcal{B}_{t,N}}Z_t-P^{(M)}_{\mathcal{B}_{t,N}}Z_{M,t}\Vert_{2} \inprob 0 \quad \text{and} \quad 
\widehat{\beta}^{(M)}_{t,N,Z_{M,t}} \inprob \beta_{t,N,Z_t} \quad \text{ as }  M \to \infty.
\end{align*}
\end{lemma}

\begin{remark}
Notice that $\widehat{V}_{N,t}^{(M)}(L)$ satisfies \eqref{eq:independence_condition} due to the backwards recursive structure of the LSM algorithm \eqref{eq:LSM_definition}.
\end{remark}

Lemma \ref{lem:simulation_stability_property} is essentially what is needed to prove the induction step in the induction argument used to prove the following result: 

\begin{theorem}\label{thm:simulation_convergence}
For $t=0,\dots,T-1$, let the mappings $\phi_t: L^2(\calH_{t+1}) \to L^2(\calH_{t})$ be $L^2$-continuous and law invariant, let $\widehat{V}_{N,t}^{(M)}(L)$ be given by \eqref{eq:LSM_definition}, and let  
$\widehat{V}_{N,t}(L)$ be given by \eqref{eq:optimal_LSM_definition}. Then, for $t = 0,\dots, T-1$, 
\begin{align*}
\Vert \widehat{V}_{N,t}(L) -\widehat{V}_{N,t}^{(M)}(L)\Vert_{2} \inprob 0 \text { as } M\to\infty. 
\end{align*}
\end{theorem}

To summarize, Lemma \ref{lem:Lp_convergence} says that we can theoretically/asymptotically achieve arbitrarily accurate approximations, even when applying the approximation recursively, and Theorem \ref{thm:simulation_convergence} says that we may approach this theoretical best approximation in practice, with enough simulated non-path-dependent data.

\begin{lemma}\label{lem:Lipschitz_implies_L2}
If the mapping $\phi_t:L^2(\calH_{t+1})\to L^2(\calH_t)$ is Lipschitz continuous in the sense that there exists a constant $K>0$ such that 
\begin{align}
|\phi_t(X)-\phi_t(Y)| \leq K\E[|X-Y|\mid \calH_t] \quad \text{for all } X,Y \in L^2(\calH_{t+1}),\label{eq:Lipschitz_cont}
\end{align}
then $\phi_t$ is $L^2$-continuous in the sense of Definition \ref{def:L2_cont}. 
\end{lemma}

\begin{lemma}\label{lem:noname1}
If the conditional monetary risk measure $\rho_t :L^2(\calH_{t+1}) \to L^2(\calH_t)$ is Lipschitz continuous in the sense of \eqref{eq:Lipschitz_cont} with Lipschitz constant $K$, then so is the mapping $\phi_t$ given by \eqref{eq:valuation_mapping_def} with Lipschitz constant $2K$.
\end{lemma}

The large class of (conditional) spectral risk measures are in fact Lipschitz continuous. These conditional monetary risk measures can be expressed as 
\begin{align}\label{def:spec_risk_meas}
\rho_{t,m}(Y) = -\int_0^1 F^{-1}_{t,Y}(u) m(u)\cald u,
\end{align}
where $m$ is a probability density function that is decreasing, bounded and right continuous, and $F^{-1}_{t,Y}(u)$ is the conditional quantile function 
\begin{align*}
F^{-1}_{t,Y}(u):=\essinf\{y\in L^0(\calH_t):\P(Y\leq y\mid \calH_t)\geq u\}.
\end{align*}
It is well known that spectral risk measures are coherent and includes expected shortfall as a special case.

\begin{lemma}\label{lem:spectral_lipschitz}
If $m$ is a probability density function that is decreasing, bounded and right continuous, then $(\rho_{t,m})_{t=0}^{T-1}$ is a dynamic monetary risk measure in the sense of Definition \ref{def:dynrisk}. Moreover, each $\rho_{t,m}$ is law invariant in the sense of Definition \ref{def:law_invariance} and also Lipschitz continuous with constant $m(0)$. 

\end{lemma}

\begin{remark}\label{rem:rate_of_conv}
Assume $\phi_t$ is Lipschitz continuous. Then
\begin{align*}
&||\widehat{V}_t(L)-V_t(L)||_2 \\
&\quad\leq ||\widehat{V}_t(L)-\phi_t(L_{t+1}+\widehat{V}_{t+1}(L))||_2 \\
&\quad\quad+ ||\phi_t(L_{t+1}+V_{t+1}(L))-\phi_t(L_{t+1}+\widehat{V}_{t+1}(L)||_2\\
&\quad \leq ||\widehat{V}_t(L)-\phi_t(L_{t+1}+\widehat{V}_{t+1}(L))||_2 +K||V_{t+1}(L)-\widehat{V}_{t+1}(L)||_2
\end{align*}
Repeating this argument gives 
\begin{align*}
&||\widehat{V}_t(L)-V_t(L)||_2 \leq \sum_{s=t}^TK^{s-t}||\widehat{V}_s(L)-\phi_s(L_{s+1}+\widehat{V}_{s+1}(L)||_2
\end{align*}
This bound is analogous to that in \cite{Zanger-13} (Lemma 2.3, see also Remark 3.4 for how this ties in with the main result) with the exception that the constant $K^{s-t}$ appears instead of $2$. As $K$ may be quite large this is one of the reasons for not seeking to determine the exact rate of convergence, as is done in for instance  \cite{Zanger-18} and \cite{Egloff-05}. This observation also discourages judging the accuracy of the LSM algorithm purely by estimating the out-of-sample one-step estimation errors of the form $||\widehat{V}_s(L)-\phi_s(L_{s+1}+\widehat{V}_{s+1}(L))||_2$, as these need to be quite small in order for a satisfying error bound.
\end{remark}

\subsection{Convergence results using Value-at-Risk}

In this section, we will focus on mappings $\phi_{t,\alpha}:L^2(\calH_{t+1}) \to L^2(\calH_t)$ given by
\begin{align}
\phi_{t,\alpha}(Y) := \VaR_{t,\alpha}(-Y) - \frac{1}{1+\eta_t}\E[(\VaR_{t,\alpha}(-Y) -Y)^+\mid \calH_t]
\label{eq:VaR_phi_def}
\end{align}
for some $\alpha \in (0,1)$ and nonnegative constants $(\eta_t)_{t=0}^{T-1}$, and where 
\begin{align*}
\VaR_{t,\alpha}(-Y)&:=F^{-1}_{t,Y}(1-\alpha)\\
&:=\essinf\{y\in L^0(\calH_t):\P(Y\leq y\mid \calH_t)\geq 1-\alpha\}.
\end{align*}
is the conditional version of Value-at-Risk. Note that $\phi_{t,\alpha}$ is a special case of mappings $\phi$ in \eqref{eq:valuation_mapping_def}. $(\VaR_{t,\alpha})_{t=0}^{T-1}$ is a dynamic monetary risk measure in the sense of Definition \ref{def:dynrisk}, and $\VaR_{t,\alpha}$ is law invariant in the sense of Definition \ref{def:law_invariance}.
Since $\VaR_{t,\alpha}$ is in general not Lipschitz continuous, $\phi_{t,\alpha}$ cannot be guaranteed to be so, without further regularity conditions. The aim of this section is to find results analogous to Lemma \ref{lem:Lp_convergence} and Theorem \ref{thm:simulation_convergence}.

We will use the following Lemma and especially its corollary in lieu of $L^2$-continuity for Value-at-Risk:
\begin{lemma}\label{lem:VaR_subadditivity}
For any $X,Z \in \calH_{t+1}$ and any $\delta \in (0, 1-\alpha)$, 
\begin{align}
&\VaR_{t,\alpha}(-(X+Z)) \leq \VaR_{t,\alpha+\delta}(-X)+\VaR_{t,1-\delta}(-Z)\label{eq:VaR_ineq_1}, \\
&\VaR_{t,\alpha}(-(X+Z)) \geq \VaR_{t,\alpha-\delta}(-X)-\VaR_{t,1-\delta}(-Z)\label{eq:VaR_ineq_2}.
\end{align}
\end{lemma}

We get an interesting corollary from this lemma:
\begin{corollary}\label{cor:VaR_Lipschitz}
Let $\alpha \in (0,1)$ and let $\delta \in (0,1-\alpha)$ with $\delta<1/2$. Then, for any $X,Y \in L^1(\calH_{t+1})$, 
\begin{align*}
\inf_{|\epsilon|<\delta}|\VaR_{t,\alpha+\epsilon}(X)-\VaR_{t,\alpha}(Y)| \leq \frac{1}{\delta}\E[|X-Y|\mid \calH_t]. 
\end{align*}
\end{corollary}

Using these Lipschitz-like results, we can show a Lipschitz-like result for $\phi_{t,\alpha}(\cdot)$. 

\begin{theorem}\label{thm:phi_Lipschitz}
Let $\alpha \in (0,1)$ and let $\delta \in (0,1-\alpha)$ with $\delta<1/2$. Then, for any $X,Y \in L^1(\calH_{t+1})$, then
\begin{align}
\inf_{|\epsilon|<\delta}|\phi_{t,\alpha+\epsilon}(X)-\phi_{t,\alpha}(Y)| \leq \frac{2}{\delta}\E[|X-Y|\mid\calH_t]
\label{eq:phi_Lipschitz}
\end{align}
and 
\begin{align}
\phi_{t,\alpha-\delta}(X) -\frac{2}{\delta}\E[|X-Y|\mid\calH_t]&\leq \phi_{t,\alpha}(Y) \nonumber \\
&\leq \phi_{t,\alpha+\delta}(X)+\frac{2}{\delta}\E[|X-Y|\mid\calH_t],  \label{eq:phi_interval}
\end{align}
and \eqref{eq:phi_Lipschitz} and \eqref{eq:phi_interval} are equivalent.
\end{theorem}

Theorem \ref{thm:phi_Lipschitz} enables us to prove $L^2$-continuity of $\phi_{t,\alpha}$ under a continuity assumption. 

\begin{corollary}\label{cor:phi_L2cont}
Consider $X,X_n\in L^2(\calH_{t+1})$, $n\geq 1$, with $X_n \to X$ in $L^2$. Assume that $(0,1) \ni u\mapsto \VaR_{t,u}(-X)$ be a.s. continuous at $u=\alpha$. Then $\phi_{t,\alpha}(X_n) \to \phi_{t,\alpha}(X)$ in $L^2$.

\end{corollary}
The following remark illustrates that even a stronger requirement of a.s. continuous time $t$-conditional distributions should not be a great hindrance in practice:
\begin{remark}
If we add to our cash flow $(L_t)_{t=1}^T$ an adapted process $(\epsilon_t)_{t=1}^T$, independent of $(L_t)_{t=1}^T$, such that for each $t$, $\epsilon_t$ is independent of $\calF_{t-1}$ and has a continuous distribution function, then the assumptions in Corollary \ref{cor:phi_L2cont} will be satisfied.
\end{remark}

We are now ready to formulate a result analogous to Lemma \ref{lem:Lp_convergence}.
\begin{lemma}\label{lem:noname2}
Let $\alpha \in (0,1)$ and let $\delta \in (0,1-\alpha)$ with $\delta<1/2$. Let $(\phi_{t,\alpha})_{t=0}^{T-1}$ be defined by \eqref{eq:VaR_phi_def} and let 
\begin{align}
\widehat{V}_{N,t,\alpha}(L) :=P_{\mathcal{B}_{t,N}}\phi_{t,\alpha}\big(L_{t+1}+\widehat{V}_{N,t+1,\alpha}(L)\big), \quad 
\widehat{V}_{N,T,\alpha}(L) :=0. 
\label{eq:optimal_VaR_LSM_definition}
\end{align}
Let $\bigcup_{n\in \N} \calB_{t,n}$ be dense in the set $\{h(S_t) \mid h:\R^d \to \R,  h(S_t) \in L^2(\calF_t)\}$ and assume that $(0,1) \ni u\mapsto \VaR_{t,u}(-L_{t+1}-\widehat{V}_{N,t+1,\alpha}(L))$ be a.s. continuous at $u=\alpha$ for all $N \in \N, t=0, \dots T-1$.
Then, for $t=0,\dots,T-1$, 
\begin{align*}
\Vert \widehat{V}_{N,t,\alpha}(L) -V_{t,\alpha}(L)\Vert_2 \in \R \text{ and }
\lim_{N\to\infty}\Vert \widehat{V}_{N,t,\alpha}(L) -V_{t,\alpha}(L)\Vert_2=0.
\end{align*}
\end{lemma}

\begin{lemma}\label{lem:VaR_L2_continuity}
Let $\alpha\in (0,1)$ and let $(0,1) \ni u\mapsto \VaR_{t,u}(v^\trans \mathbf{\Phi}_{t+1,N})$ be a.s. continuous at $u=\alpha$ for any $v \in \R^N$. Then 
\begin{align*}
\beta_n \inprob \beta \quad \text{implies} \quad 
\big\Vert\phi_{t,\alpha}(\beta^\trans \mathbf{\Phi}_{t+1,N})-\phi_{t,\alpha}(\beta_n^\trans \mathbf{\Phi}_{t+1,N})\big\Vert_ 2 \inprob 0.
\end{align*}
\end{lemma}

\begin{remark}
Lemma \ref{lem:VaR_L2_continuity} can be extended to show the convergence  
\begin{align*}
\big\Vert \phi_{t,\alpha}(L_{t+1}+\beta^\trans \mathbf{\Phi}_{t+1,N})-\phi_{t,\alpha}(L_{t+1}+\beta_n^\trans \mathbf{\Phi}_{t+1,N})\big\Vert_ 2 \inprob 0
\end{align*}
since the vector of basis functions $\mathbf{\Phi}_{t+1,N}$ could contain $L_{t+1}$ as an element. The requirement for convergence is that $u\mapsto \VaR_{t,u}(-L_{t+1}-v^\trans \mathbf{\Phi}_{t+1,N})$ is a.s. continuous at $u=\alpha$.
This requirement could be replaced by the stronger requirement that $x\mapsto \P(L_{t+1}+v^\trans \mathbf{\Phi}_{t+1,N}\mid\calF_t)$ is a.s. continuous.
\end{remark}

We have now fitted $\phi_{t,\alpha}$ into the setting of Theorem \ref{thm:simulation_convergence}.

\begin{theorem}\label{thm:noname3}
Let $u\mapsto \VaR_{t,u}(-L_{t+1}-v^\trans \mathbf{\Phi}_{t+1,N})$ be a.s. continuous at $u=\alpha$ for any $v \in \R^N$. For any $N\in N$ and $t=0,\dots,T$, let $\widehat{V}_{N,t,\alpha}(L)$ be given by \eqref{eq:optimal_VaR_LSM_definition} and define
\begin{align*}
\widehat{V}^{(M)}_{N,t,\alpha}(L) :=P^{(M)}_{\mathcal{B}_{t,N}}\phi_{t,\alpha}\big(L_{t+1}+\widehat{V}^{(M)}_{N,t+1}(L)\big), 
\quad \widehat{V}^{(M)}_{N,T,\alpha}(L) :=0.
\end{align*}
Then, for $t=0,1,\dots,T-1$, $\Vert\widehat{V}_{N,t,\alpha}(L) -\widehat{V}_{N,t,\alpha}^{(M)}(L)\Vert_{2} \inprob 0$ as $M\to\infty$.
\end{theorem}

\section{Implementing and validating the LSM algorithm}\label{sec:numerical_results}

In this section we will test the LSM algorithm empirically for the special case of the mappings $\phi$ being given by $(\phi_{t,\alpha})_{t=0}^{T-1}$ in \eqref{eq:VaR_phi_def}. The LSM algorithm described below, Algorithm \ref{alg:LSM}, will differ slightly from the one previously, in the sense that it will contain the small inefficiency of having two regression steps: One for the $\VaR$ term of the mapping, one for the expected value term. The reason for introducing this split is that it will significantly simplify the validation procedures of the algorithm. Heuristically, we will be able to run a forward simulation where we may test the accuracy of both the $\VaR$ term and the expected value term. 

Let $\mu_{t,t+1}(\cdot,\cdot)$ be the transition kernel from time $t$ to $t+1$ of the Markov process $(S_t)_{t=0}^T$ so that $\mu_{t,t+1}(S_t,\cdot)=\P(S_{t+1}\in \cdot \mid S_t)$.
In order to perform the LSM algorithm below, the only requirements are the ability to efficiently sample a variate $s$ from the unconditional law $\mathcal{L}(S_t)$ of $S_t$ and from the conditional law $\mu_{t,t+1}(\cdot,s)$. 
Recall that the liability cash flow $(L_t)_{t=1}^T$ is assumed to be given by $L_t := g_t(S_t)$, for known functions $(g_t)_{t=1}^T$. 

\begin{algorithm}
\caption{LSM Algorithm}
\begin{algorithmic} \label{alg:LSM}
\STATE{Set $\widehat{\beta}^{(M)}_{T,N,V}:=0$}
\FOR{$t=T-1:0$}
\STATE{Draw independent variables $S^{(1)}_t, \dots S^{(M)}_t$} from $\mathcal{L}(S_t)$
\FOR{$i=1:M$} 
\STATE {Draw independent variables $S^{(i,1)}_{t+1}, \dots, S^{(i,n)}_{t+1}$ from $\mu_{t,t+1}(S^{(i)}_t,\cdot)$}
\STATE {Set $Y^{(i,j)}_{t+1}:=g_{t+1}(S^{(i,j)}_{t+1})+(\widehat{\beta}^{(M)}_{t+1,N,V})^{\trans}\mathbf{\Phi}_{t+1,N}(S^{(i,j)}_{t+1})$, $j=1,\dots,n$}
\STATE{Let $\widehat{F}^{(i)}_t(y):=\frac{1}{n}\sum_{j=1}^n I\{Y^{(i,j)}_{t+1}\leq y\}$ (empirical cdf)}
\STATE{Set $R^{(i)}_t:=\min\{y:\widehat{F}^{(i)}_t(y)\geq \alpha\}$ (empirical $\alpha$-quantile)}
\STATE{Set $E^{(i)}_t:=\frac{1}{n}\sum_{j=1}^n (R^{(i)}_t-Y^{(i,j)}_{t+1})_+$}
\ENDFOR
\STATE{Set $\widehat{\beta}^{(M)}_{t,N,R}$ as in  
\eqref{eq:beta_estimate_def} by regressing $(R_t^{(i)})_{i=1}^M$ onto $(\mathbf{\Phi}_{t,N}(S^{(i)}_{t}))_{i=1}^M$}
\STATE{Set $\widehat{\beta}^{(M)}_{t,N,E}$ as in  
\eqref{eq:beta_estimate_def} by regressing $(E_t^{(i)})_{i=1}^M$ onto $(\mathbf{\Phi}_{t,N}(S^{(i)}_{t}))_{i=1}^M$}
\STATE{Set $\widehat{\beta}^{(M)}_{t,N,V}:=\widehat{\beta}^{(M)}_{t,N,R_t}-\frac{1}{1+\eta}\widehat{\beta}^{(M)}_{t,N,E_t}$}
\ENDFOR
\end{algorithmic}
\end{algorithm}

We may assess the accuracy of the LSM implementation by computing root mean-squared errors (RMSE) of quantities appearing in Algorithm \ref{alg:LSM}. For each index pair $(t,i)$ set $V^{(i)}_t:=R^{(i)}_t-\frac{1}{1+\eta}E^{(i)}_t$. Define the RMSE and the normalized RMSE by 
\begin{align}
\text{RMSE}_{Z,t} &:= \bigg(\frac{1}{M}\sum_{i=1}^M\Big(Z^{(i)}_{t}-(\widehat{\beta}^{(M)}_{t,N,Z})^\trans\mathbf{\Phi}_{t,N}\big(S^{(i)}_{t}\big)\Big)^2\bigg)^{1/2}, \label{eq:RMSE}\\
\text{NRMSE}_{Z,t} &:= \text{RMSE}_{Z,t} \times \bigg(\frac{1}{M}\sum_{i=1}^M {Z^{(i)}_{t}}^2\bigg)^{-1/2}\label{eq:NRMSE},
\end{align}
where $Z$ is a placeholder for $R$, $E$ or $V$.

For each index pair $(t,i)$ consider the actual non-default probability and actual return on capital given by 
\begin{align}
\text{ANDP}^{(i)}_t&:=\widehat{F}^{(i)}_t\big((\widehat{\beta}^{(M)}_{t,N,R})^{\trans}\mathbf{\Phi}_{t,N}(S^{(i)}_{t}))\big),\label{eq:ANDP}\\
\text{AROC}^{(i)}_t&:=(1+\eta)E^{(i)}_t\times \big((\widehat{\beta}^{(M)}_{t,N,E})^{\trans}\mathbf{\Phi}_{t,N}(S^{(i)}_{t})\big)^{-1}),\label{eq:AROC}
\end{align}
and note that these random variables are expected to be centered around $\alpha$ and $1+\eta$, respectively, if the implementation is accurate. All validation procedures in this paper are performed out-of-sample, i.e. we must perform a second validation run to get these values. 

\subsection{Models}

In this section we will introduce two model types in order to test the performance of the LSM algorithm. The first model type, introduced in Section \ref{sec:GARCH}, is not motivated by a specific application but is simply a sufficiently flexible and moderately complex time series model.The second model type, introduced in Section \ref{sec:life_model}, aims to describe the cash flow of a life insurance portfolio paying both survival and death benefits.   

\subsubsection{AR(1)-GARCH(1,1) models}\label{sec:GARCH}

The first model to be evaluated is when the liability cash flow $(L_t)_{t=1}^T$ is assumed to be given by a process given by an AR(1) model with GARCH(1,1) residuals, with dynamics given by:
\begin{align*}
L_{t+1} = \alpha_0 + \alpha_1L_{t} + \sigma_{t+1}\epsilon_{t+1}, \quad
\sigma^2_{t+1} =\alpha_2 +\alpha_3\sigma^2_{t}+ \alpha_4L^2_{t}, \quad
L_0=0, \sigma_1=1.
\end{align*}
Here $\epsilon_1, \dots, \epsilon_T$ are assumed to be i.i.d. standard normally distributed and $\alpha_0, \dots \alpha_4$ are known model parameters. If we put $S_t = (L_t,\sigma_{t+1})$ for $t=0,\dots, T$, we see that $S_t$ will form a time homogeneous Markov chain. 

In order to contrast this model with a more complex model, we also investigate the case where the process $(L_t)_{t=1}^T$ is given by a sum of independent AR(1)-GARCH(1,1)-processes of the above type: $L_t = \sum_{i=1}^{10} L_{t,i}$, where
\begin{align*}
L_{t+1,i} = \alpha_{0,i} + \alpha_{1,i}L_{t,i} + \sigma_{t+1,i}\epsilon_{t+1,i}, \quad
\sigma^2_{t+1,i} =\alpha_{2,i} +\alpha_{3,i}\sigma^2_{t,i}+ \alpha_{4,i}L^2_{t,i}.
\end{align*}
The motivation for these choices of toy models is as follows: Firstly, a single AR(1)-GARCH(1,1) process is sufficiently low dimensional so we may compare brute force approximation with that of the LSM model, thus getting a real sense of the performance of the LSM model. Secondly, despite it being low dimensional, it still seems to have a sufficiently complex dependence structure as not to be easily valued other than by numerical means. 
The motivation for looking at a sum of AR(1)-GARCH(1,1) processes is simply to investigate whether model performance is severely hampered by an increase in dimensionality, provided a certain amount of independence of the sources of randomness. 

\subsubsection{Life insurance models}\label{sec:life_model}

In order to investigate a set of models more closely resembling an insurance cash flow, we also consider an example closely inspired by that in \cite{Delong-19}. Essentially, we will assume the liability cash flow to be given by life insurance policies where we take into account age cohorts and their sizes at each time, along with financial data relevant to the contract payouts.

We consider two risky assets $Y$ and $F$, given by the log-normal dynamics
\begin{align*}
&\cald Y_t = \mu_Y Y_t \cald t + \sigma_Y Y_t \cald W^{Y}_t,\quad 0 \leq t\leq T,  \quad Y_0 =y_0,\\
&\cald F_t = \mu_F F_t \cald t + \sigma_F F_t \cald W^{F}_t, \quad 0 \leq t\leq T,\quad F_0 =f_0.
\end{align*}
$W^{Y}_t$ $ W^{F}_t$ are two correlated Brownian motions, which we may re-write as
\begin{align*}
&W^{Y}_t = W^1_t, \quad W^F_t = \rho W^1_t + \sqrt{1 - \rho^2}W^2_t, \quad 0 \leq t\leq T,
\end{align*}
where $W^1$ and $W^2$ are two standard, uncorrelated Brownian motions. Here, $F$ will represent the index associated with unit-linked contracts and  $Y$ will represent assets owned by the insurance company. Furthermore, we assume that an individual of age $a$ has the probability $1-p_a$ of reaching age $a+1$, where the probabilities $p_a$ for $a=0,1, \dots $ are assumed to be nonrandom and known. All deaths are assumed to be independent of each other. We will consider $k$ age-homogeneous cohorts of sizes $n_1, \dots, n_k$ at time $t=0$ and ages $a_1, \dots, a_k$ at time $t=0$. We assume that all insured individuals have bought identical contracts. If death occurs at time $t$, the contract pays out the death benefit $\max(D^*,F_t)$, where $D^*$ is a nonrandom guaranteed amount. If an insured person survives until time $T$, the survival benefit $\max(S^*,F_T)$, where again $S^*$ is a nonrandom amount. We finally assume that the insurance company holds the nominal amount $c(n_1+ \dots +n_k)$ in the risky asset Y and that they will sell off these assets proportionally to the amount of deaths as they occur, and sell off the entire remaining amount at time $T$. Let $N^i_t$ denote the number of people alive in cohort $i$ at time $t$, with the following dynamics:
\begin{align*}
&N^i_{t+1} \sim \text{Bin}(N^i_{t}, 1-p_{a_i+t}), \quad t=0, \dots, T-1.
\end{align*}
These are the same dynamics as the life insurance example in Section 5 of \cite{Engsner-Lindholm-Lindskog-17}. Thus, the liability cash flow we consider here is given by
\begin{align*}
L_t &= \big(\max(D^*,F_t) - cY_t\big)\sum_{i=1}^k(N^i_{t}-N^i_{t-1})\\
&\quad + \calI\{t=T\}\big(\max(S^*,F_T) - cY_T\big)\sum_{i=1}^k N^i_{T} 
\end{align*}
If we write $S_t = (Y_t, F_t, N^1_t, \dots, N^k_t)$, then $S:=(S_t)_{t=0}^T$ will be a Markov chain with dynamics outlined above. Note that depending on the number $k$ of cohorts, $S$ might be a fairly high-dimensional Markov chain. Note that in addition to the obvious risk factors of mortality and the contractual payout amounts, there is also the risk of the value of the insurance company's risky asset $Y$ depreciating in value, something which is of course a large risk factor of insurance companies in practice. Here we will consider the case of $k=4$ cohorts, referred to as the small life insurance model and the case $k=10$ cohorts, referred to as the large life insurance model. 

\subsection{Choice of basis functions}

So far, the choice of basis functions has not been addressed. As we are trying to numerically calculate some unknown functions we do not know the form of, the approach used here will be a combination of standard polynomial functions, completed with functions that in some ways bear resemblance to the underlying liability cash flow. A similar approach for the valuation of American derivatives is taken in for instance \cite{Longstaff-Schwartz-01} and \cite{Brodie-15}, where in the latter it is explicitly advised (see pp.~1082) to use the value of related, simpler, derivatives as  basis functions to price more exotic ones.

In these examples, we will not be overly concerned with model sparsity, covariate significance or efficiency, but rather take the machine-learning approach of simply evaluating models based on out-of-sample performance. This is feasible due to the availability of simulated data for both fitting and out-of-sample validation.

\subsubsection{AR(1)-GARCH(1,1) models}

Since the AR(1)-GARCH(1,1) models can be considered toy models, generic basis functions were chosen. For a single AR(1)-GARCH(1,1) model, the choice of basis functions was all polynomials of the form $L_{t}^i\sigma_{t+1}^j$ for all $0 < i+j \leq 2$.
For the sum of $10$ independent AR(1)-GARCH(1,1) models we denote by $L_t, \sigma_{t+1}$ the aggregated liability cash flow and standard deviation at time $t$ and $t+1$, respectively. Then we consider the basis functions consisting of the state vector $(L_{t,i},  \sigma_{t,i})_{i=1}^{10}$ along with $L_{t}^i\sigma_{t+1}^j$ for all $0 < i+j \leq 2$, omitting the case of $i=1$, $j=0$ to avoid collinearity. Note that the number of basis functions grow linearly with the dimensionality of the state space, rather than quadratically. 

\subsubsection{Life insurance models}

For the state $S_t = (Y_t, F_t, N^1_t, \dots, N^k_t)$, let $p^i_{t+1}$ be the probability of death during $(t, t+1)$ for an individual in cohort $i$, with $q^i_{t+1}:= 1- p^i_{t+1}$. We then introduce the state-dependent variables
\begin{align*}
\mu_{t+1} := \sum_{i=1}^k  N^i_tp^i_{t+1}, \quad
\sigma_{t+1} := \Big(\sum_{i=1}^k  N^i_tp^i_{t+1}q^i_{t+1} \Big)^{1/2}, \quad
N_{t} := \sum_{i=1}^k  N^i_t.
\end{align*}
The first two terms here are the mean and standard deviation of the number of deaths during $(t, t+1)$, the third simply being the total number of people alive at time $t$. The basis functions we choose consist of the state vector $Y_t, F_t, N^1_t, \dots, N^k_t$ together with all products of two factors where the first factor is an element of the set $\{\mu_{t+1}, \sigma_{t+1}, N_{t}\}$ and the other factor is an element of the set
\begin{align*}
&\big\lbrace Y_t, F_t, Y^2_t, F^2_t , F_t^3, Y_tF_t, Y_tF_ t^2,(F_t - K_j)_+, (F_t - K_j)_+Y_t, \\
&\quad C(F_t,S^*  ,T,t) , C(F_t,D^*,t+1,t) , C(F_t,S^*  ,T,t)Y_t , C(F_t,D^*,t+1,t)Y_t\big\}.
\end{align*}
$K_j$ can take values in $\{200,162,124,103\}$ depending on which covariates of the form $(F_t - K_j)_+$ had the highest $R^2$-value at time $T=5$. Here the $R^2$-values were calculated based on the residuals after performing linear regression with respect to all basis function not containing elements of the form $(F_t - K_j)_+$. While this is a somewhat ad hoc approach that could be refined, it is a simple and easy to implement example of basis functions. Again note that the number of basis functions grow linearly with the dimensionality of the state space, rather than quadratically.

\subsubsection{Run specifications}\label{seq:run_specs}

For Algorithm \ref{alg:LSM}, $M=5\cdot 10^4$ and $n = 10^5$ were chosen for the life insurance models and $M=10^4$ and  $n = 10^5$ for the AR(1)-GARCH(1,1) models. Terminal time $T = 6$ was used in all cases. For the validation run, $M=10^4$ and $n = 10^5$ were chosen for all models. Due to the extreme quantile level involved, and also based on empirical observations, it was deemed necessary to keep $n$ around this order of magnitude. Similarly, in part due to the number of basis functions involved, it was observed as well that performance seemed to increase with $M$. The choice of $M$ and $n$ to be on the considered order of magnitude was thus necessary for good model performance, and also the largest orders of magnitude that was computationally feasible given the computing power available.

For the AR(1)-GARCH(1,1) model, the chosen parameters were 
$$\alpha_0 =1 , \quad \alpha_1 =1, \quad \alpha_2 =0.1, \quad \alpha_3 =0.1, \quad \alpha_4 =0.1.$$ 
The same choice was used for each of the terms in the sum of $10$ AR(1)-GARCH(1,1) processes, making the model a sum of i.i.d. processes. 

For the life insurance models, the choice of parameters of the risky assets was $\mu_Y=\mu_F=0.03, \sigma_Y=\sigma_F =0.1, \rho=0.4, y_0=f_0 =100$. The benefit lower bounds were chosen as $D^* = 100, S^* =110$. The death/survival probabilities were calculated using the Makeham formula (for males):
\begin{align*}
p_a = \exp \Big\{-\int_{a}^{a+1} \mu_x \cald x\Big\} \quad
\mu_x :=0.001 + 0.000012\exp\{0.101314 x\}.
\end{align*}
These numbers correspond to the Swedish mortality table M90 for males (the formula for females is identical, but adjusted backwards by $6$ years to account for the greater longevity in the female population). For the case of $4$ cohorts, starting ages (for males) were $50-80$ in $10$-year increments and for the case of $10$ cohorts the starting ages were $40-85$ with $5$-year increments.

The algorithms were run on a computer with 8 Intel(R) Core(TM) i7-4770S  3.10GHz processors, and parallel  programming was implemented in the nested simulation steps in both Algorithm \ref{alg:LSM} and the validation algorithm.

\subsection{Numerical results}

The RMSE:s and NRMSE:s of the LSM models can be seen in Table \ref{tab:tab1} (RMSE:s) and Table \ref{tab:tab2} (NRMSE:s). 
The ANDP:s and AROC:s of the LSM models can be seen in Table \ref{tab:tab3}.
The tables display quantile ranges with respect to the $2.5 \%$ and $97.5 \%$ quantiles of the data. 

\begin{table}
\begin{center}
    \begin{tabular}{ | p{3.5cm} | p{2.5cm} | p{2.5cm} | p{2.5cm} |}
    \hline
    Model & RMSE V & RMSE R & RMSE E  \\ \hline
    \text{one single} \text{AR(1)-GARCH(1,1)}  & 0.0114, 0.0118, 0.0115, 0.0098, 0.0061 & 0.0533, 0.0556, 0.0553, 0.0455, 0.0285 & 0.0521, 0.0542, 0.0544, 0.0444, 0.0279 \\ \hline
    \text{a sum of $10$} \text{AR(1)-GARCH(1,1)}  & 0.0172, 0.0130, 0.0120, 0.0100, 0.0061 & 0.0525, 0.0552, 0.0546, 0.0467, 0.0278 & 0.0536, 0.0544, 0.0535, 0.0458, 0.0273 \\ \hline
    \text{Life model} \text{with 4 cohorts} & 134.4, 120.3, 134.8,  85.1,  75.5& 760.3, 682.7, 901.6, 535.9, 575.1 & 742.4, 665.0, 856.8, 536.0, 571.3 \\ \hline
    \text{Life model} \text{with 10 cohorts} & 331.9, 307.4, 330.8, 226.2, 219.3 & 1730.1, 1719.4, 2148.3, 1431.0, 1928.3 & 1689.2, 1672.8, 2049.8, 1429.4, 1910.3 \\ \hline
    \end{tabular}
    \caption{RMSE values for the quantities $V, R, E$ as defined in \eqref{eq:RMSE}. The five values in each cell are for times $t=1, 2, 3, 4, 5$, in that order.}
    \label{tab:tab1}
\end{center}
\end{table}

\begin{table}
\begin{center}
    \begin{tabular}{ | p{3.5cm}  | p{2.5cm} | p{2.5cm}  |p{2.5cm} |}
    \hline
    Model & NRMSE V (\%)& NRMSE R (\%)& NRMSE E (\%) \\ \hline
    \text{one single} \text{AR(1)-GARCH(1,1)}  & 0.0498, 0.0583, 0.0685, 0.0797, 0.0901 &0.1705, 0.1931, 0.2170, 0.2333, 0.2544 & 0.5810, 0.5962, 0.5930, 0.5747, 0.5885 \\ \hline
    \text{a sum of $10$} \text{AR(1)-GARCH(1,1)} &0.0758, 0.0642, 0.0715, 0.0813, 0.0912&  0.1693, 0.1913, 0.2158, 0.2393, 0.2493 & 0.6049, 0.5963, 0.5880, 0.5949, 0.5779 \\ \hline
    \text{Life model} \text{with 4 cohorts} &  0.2567 0.2225 0.2603 0.1711 0.1647 &  0.5443 0.5646 0.8722 0.5924 0.7403 & 0.7109 0.7535 1.1381 0.8306 1.0271 \\ \hline
    \text{Life model} \text{with 10 cohorts} & 0.2505, 0.2247, 0.2454, 0.1720, 0.1733 &  0.4911, 0.5592, 0.7948, 0.5952, 0.9056 &  0.6475, 0.7452, 1.0499, 0.8351, 1.2580 \\ \hline
    \end{tabular}
    \caption{NRMSE values for the quantities $V, R, E$ as defined in \eqref{eq:NRMSE}. The five values in each cell are for times $t=1, 2, 3, 4, 5$, in that order.}
    \label{tab:tab2}    
\end{center}
\end{table}

\begin{table}
\begin{center}
    \begin{tabular}{ | p{3.5cm}  | p{2.5cm} | p{2.5cm} |}
    \hline
    Model & QR  ANDP ($2.5 \%$, $97.5 \%$) & QR AROC ($2.5 \%$, $97.5 \%$) \\ \hline
    \text{one single} \text{AR(1)-GARCH(1,1)} & (0.457, 0.544), (0.456, 0.545),  (0.457, 0.545), (0.458, 0.545), (0.457, 0.543) & (4.79, 7.22), (4.76, 7.25), (4.78, 7.22), (4.84, 7.24), (4.80, 7.20) \\ \hline
    \text{a sum of $10$} \text{AR(1)-GARCH(1,1)} &  (0.458, 0.545), (0.456, 0.545), (0.457, 0.545), (0.456, 0.546), (0.457, 0.544) & (4.74, 7.29), (4.77, 7.26), (4.77, 7.23), (4.79, 7.25), (4.81, 7.21) \\ \hline
    \text{Life model} \text{with 4 cohorts} &  (0.454, 0.548), (0.443, 0.565), (0.385, 0.622), (0.436, 0.571), (0.387, 0.603)  &  (4.59, 7.43), (4.32, 7.82), (2.71, 9.05), (4.10, 7.93), (2.69, 8.49)  \\ \hline
    \text{Life model} \text{with 10 cohorts} &  (0.457, 0.546), (0.444, 0.560), (0.391, 0.611), (0.435, 0.569), (0.394, 0.605) & (4.66, 7.37), (4.33, 7.68), (2.94, 8.97), (4.08, 7.89), (2.96, 8.58)  \\ \hline
    \end{tabular}
    \caption{ Quantile ranges for the samples $(1-\text{ANDP}_{t}^{(i)})_{i=1}^{M}$ and $(\text{AROC}_{t}^{(i)})_{i=1}^{M}$, as defined in \eqref{eq:ANDP} and \eqref{eq:AROC}. The quantiles considered are $2.5 \%$ and $97.5\%$. The five intervals in each cell are for times $t=1, 2, 3, 4, 5$, in that order.}
    \label{tab:tab3}    
\end{center}
\end{table}

Below, in Figure \ref{fig:histograms} we also present some histograms of the actual returns and risks of ruin, in order to get a sense of the spread of these values.

\begin{figure}[htp]
  \centering
  \begin{tabular}{cc}

    \includegraphics[width=60mm]{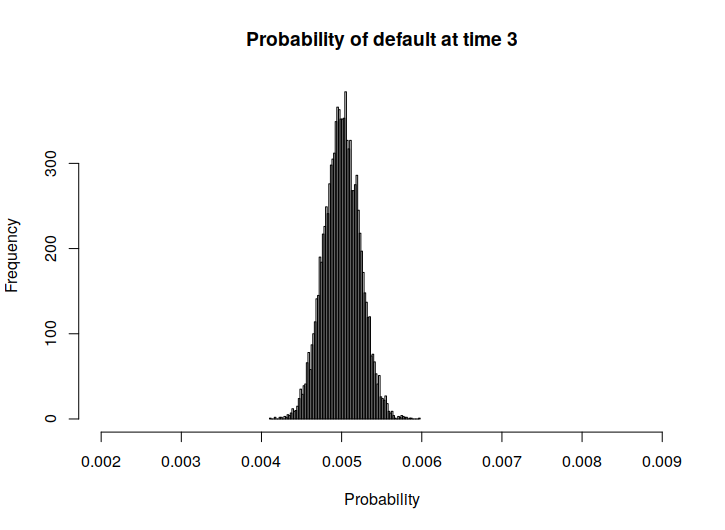}&

    \includegraphics[width=60mm]{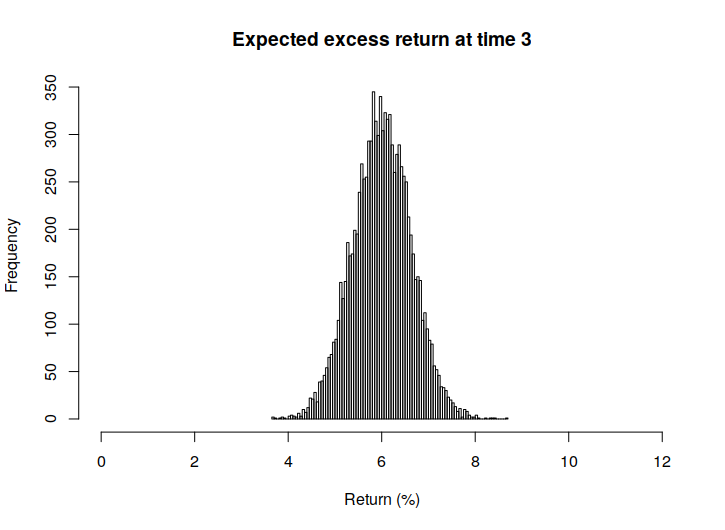}\\

    \includegraphics[width=60mm]{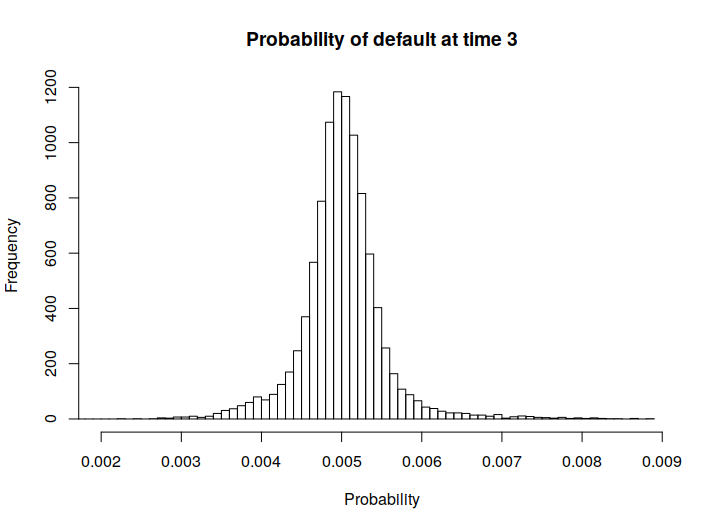}&

    \includegraphics[width=60mm]{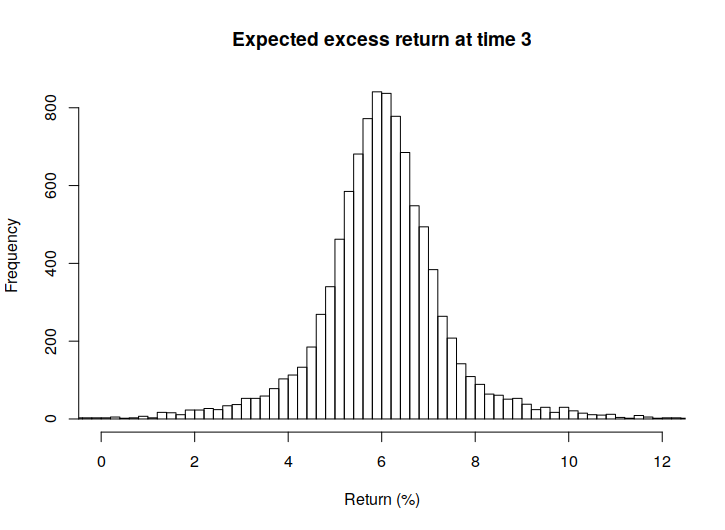}\\

  \end{tabular}
\caption{The top two figures correspond to the AR(1)-GARCH(1,1) model. The bottom two figures correspond to the life insurance model with 10 cohorts.}
\end{figure}\label{fig:histograms}

From these we can observe that the quantity representing the actual returns seems to be quite sensitive to model errors, if we recall the rather small size of the RMSE values.

\begin{table}
\begin{center}
    \begin{tabular}{ | p{3.5cm}  | p{2.5cm} | p{2.5cm} |}
    \hline
    Model & Running time valuation (HH:MM)& Running time validation (HH:MM) \\ \hline
    \text{one single} \text{AR(1)-GARCH(1,1)} & 00:06 & 00:10 \\ \hline
    \text{a sum of 10} \text{AR(1)-GARCH(1,1)}  & 00:33 &  00:39 \\ \hline
    \text{Life model} \text{with 4 cohorts} &  12:48  &  02:30 \\ \hline
    \text{Life model} \text{with 10 cohorts} &  13:29 & 02:44 \\ \hline
    \end{tabular}
    \caption{Run time of each model in hours and minutes. Run specifications are described in section \ref{seq:run_specs}}
    \label{tab:running_times}
\end{center}
\end{table}

Table \ref{tab:running_times} displays the running times of each model. As far as is known, the main factor determining running time of Algoritm \ref{alg:LSM} is the repeated calculation of the basis functions inside the nested simulation (required to calculate the quantities $Y_{t+1}^{i,j}$ in the inner for-loop). As these are quite many for models with high-dimensional state spaces, we see that running times increase accordingly. It should be noted that Algorithm \ref{alg:LSM} was not implemented to run as fast as possible for any specific model, other than the implementation of parallel programming. Speed could potentially be gained by adapting Algorithm \ref{alg:LSM} for specific models of interest.

Some conclusions can be drawn from the numerical results. Firstly, we can see that from a mean-squared-error point of view, the LSM model seems to work well in order to capture the dynamics of the multiperiod cost-of-capital valuation. It should be noted that the (N)RMSE of the value $V$ is lower than those of $R$ and $E$ across the board for all models and times. Since the expression for $E$ is heavily dependent of $R$, we can suspect that estimation errors of $R$ and $E$ are positively correlated, and thus that $V = R -\frac{1}{1+\eta}E$ gets lower mean squared errors as a result.

We can see that increasing model complexity for the AR(1)-GARCH(1,1) and life insurance models seems to have no significant effect on LSM performance. It should be noted that model complexity in both cases is increased by introducing independent stochastic factors; A sum of i.i.d. processes in the AR(1)-GARCH(1,1) case and the adding of (independent) cohorts in the life insurance case. Thus the de-facto model complexity might not have increased much, even though the state-space of the markov process is increased. 

When we look at the ANDP and AROC quantities, we see that these seem to vary more than do the (N)RMSE:s. Especially AROC, which is defined via a quotient, seems to be sensitive to model error.

One important thing to note with regards to sensitivity of ANDP and AROC is the presence of errors introduced by the necessity of having to use Monte-Carlo simulations in order to calculate samples of $\VaR_{t,1-\alpha}(-\cdot), \E[(\cdot)_+]$. This can be seen in the AR(1)-GARCH(1,1) case: If we investigate what the value $V_5(L)$ should be, we see that in this case it has a closed form (using positive homogeneity and translation invariance):
\begin{align*}
\phi_{5}(L_6 +V_6(L)) = \phi_{5}(\alpha_0 + \alpha_1L_5 + \sigma_6\epsilon_6) =  \alpha_0 + \alpha_1L_5 +\sigma_6\phi_{5}(\epsilon_6).
\end{align*}
$\phi_{5}(\epsilon_6)$ is deterministic due to law invariance. Since $L_t$ and $\sigma_t$ are included in the basis functions for the AR(1)-GARCH(1,1) model, we would expect the fit in this case to be perfect. Since it is not, we conclude that errors still may appear even if optimal basis functions are among our selection of basis functions. 

Finally, if we recall that the main purpose of these calculations is to calculate the quantity $V_t(L)$, a good approach for validation might be to re-balance the LSM estimates of $R^{(i)}_t$ and $E^{(i)}_t$ so that the LSM estimate of the value $V_t$ remains unchanged, but the LSM estimates are better fitted to $R^{(i)}_t$ and $E^{(i)}_t$. This re-balancing would not be problematic in the economic model that this validation scheme is played out in. However, in this paper we were also interested in how the LSM model captures both the VaR term and the expected value term, so the quantities ANDP and AROC remain relevant to look at.

\section{Conclusion}

We have studied the performance of the LSM algorithm to numerically compute recursively defined objects such as  $(V_t(L,\phi))_{t=0}^T$ given in definition \ref{def:phi_value}, where the mappings $\phi$ are either $L^2$-continuous or are given by \eqref{eq:VaR_phi_def}. As a part of this study, Lipschitz-like results and conditions for $L^2$-continuity were established for Value-at-Risk and the associated operator $\phi_{t,\alpha}$ in Theorem \ref{thm:phi_Lipschitz} and Corollary \ref{cor:phi_L2cont}. Important basic consistency results have been obtained showing the convergence of the LSM estimator both as the number of basis functions go to infinity in Lemmas \ref{lem:Lp_convergence} and \ref{lem:noname2} and when the size of the simulated data goes to infinity for a fixed number of basis functions in Theorems \ref{thm:simulation_convergence} and \ref{thm:noname3}. Furthermore, these results are applicable to a large class of conditional monetary risk measures, utility functions and various actuarial multi-period valuations, the only requirement being $L^2$-continuity or a property like that established in Theorem \ref{thm:phi_Lipschitz}.
We also apply and evaluate the LSM algorithm with respect to multi-period cost-of-capital valuation considered in \cite{Engsner-Lindholm-Lindskog-17} and \cite{Engsner-Lindskog-20}, and in doing this also provide insight into practical considerations concerning implementation and validation of the LSM algorithm. 
\section{Proofs}

\begin{proof}[Proof of Lemma \ref{lem:Lp_convergence}]
Note that the quantities defined in \eqref{eq:phivalue_def} and \eqref{eq:optimal_LSM_definition} are independent of $D$, hence all norms below are a.s. constants. Define $\epsilon_t :=  \widehat{V}_{N,t}-V_t$. We will now show via backwards induction staring from time $t=T$ that $||\epsilon_t||_{2}  \to 0$. The induction base is trivial, since $\widehat{V}_{N,T}(L) = V_T(L)=0$. Now assume that $||\epsilon_{t+1}||_{2} \to 0$. Then
\begin{align*}
||\epsilon_t||_{2} &\leq 
||\widehat{V}_{N,t} - \phi_t(L_{t+1}+\widehat{V}_{N,t+1}(L))||_{2}\\
&\quad+ ||\phi_t(L_{t+1}+\widehat{V}_{N,t+1}(L))-\phi_t(L_{t+1}+V_{t+1}(L))||_{2}
\end{align*}
By the induction assumption and the continuity assumption, we know that the second summand goes to $0$. We now need to show that $||\widehat{V}_{N,t} - \phi_t(L_{t+1}+\widehat{V}_{N,t+1}(L))||_{2} \to 0$. Now we simply note, by the definition of the projection operator and denseness of the approximating sets,
\begin{align*}
&||\widehat{V}_{N,t} - \phi_t(L_{t+1}+\widehat{V}_{N,t+1}(L))||_{2}\\
&\quad = \inf_{B \in \mathcal{B}_{N}}||B-\phi_t(L_{t+1}+\widehat{V}_{N,t+1}(L))||_{2}\\
&\quad \leq  \inf_{B \in \mathcal{B}_{N}}||B-\phi_t(L_{t+1}+V_{t+1})||_{2}\\
&\quad\quad +||\phi_t(L_{t+1}+\widehat{V}_{N,t+1}(L))-\phi_t(L_{t+1}+V_{t+1})||_{2}\\
&\quad =\Big[   \inf_{B \in \mathcal{B}_{N}}||B-\phi_t(L_{t+1}+V_{t+1})||_{2}\Big]\\
&\quad\quad + ||\phi_t(L_{t+1}+\widehat{V}_{N,t+1}(L))-\phi_t(L_{t+1}+V_{t+1})||_{2}.
\end{align*}
By our assumptions, both these terms go to zero as $\phi_t(L_{t+1}+V_{t+1}(L))$ is afunction of the state $S_t$, which lies in $L^2(\calF_t)$.
\end{proof}

\begin{proof}[Proof of Lemma \ref{lem:simulation_stability_property}]
We first note that if $\widehat{\beta}^{(M)}_{t,N,Z_{M,t}} \to \beta_{t,N,Z_t}$ in probability, then $||(\beta_{t,N,Z_t})^\trans\mathbf{\Phi}_{t,N} -(\widehat{\beta}^{(M)}_{t,N,Z_{M,t}})^\trans\mathbf{\Phi}_{t,N} ||_{2} \to 0$ in probability, since $\widehat{\beta}^{(M)}_{t,N,Z_{M,t}}$ is independent of $\mathbf{\Phi}_{t,N}$ and $\calF_0$-measurable, while $\mathbf{\Phi}_{t,N}$ is independent of $D$. Hence it suffices to show that $\widehat{\beta}^{(M)}_{t,N,Z_{M,t}} \to \beta_{t,N,Z_t}$ in probability. Now, recalling the definition of $\widehat{\beta}^{(M)}_{t,N,Z_{M,t}}$ , we re-write \eqref{eq:beta_estimate_def} as
\begin{align*}
 &\widehat{\beta}^{(M)}_{t,N,Z_t} = \Big(\frac{1}{M}\big(\mathbf{\Phi}^{(M)}_{t,N}\big)^\trans \mathbf{\Phi}^{(M)}_{t,N} \Big)^{-1}\frac{1}{M}\big(\mathbf{\Phi}^{(M)}_{t,N}\big)^\trans Z_t^{(M)},
\end{align*}
Furthermore recall the form of $\beta_{t,N,Z_t}$ given by \eqref{eq:opt_beta_def}. We first note that since, by the law of large numbers, $\frac{1}{M}\big(\mathbf{\Phi}^{(M)}_{t,N}\big)^\trans \mathbf{\Phi}^{(M)}_{t,N} \to \E_0\big[\mathbf{\Phi}_{t,N} \mathbf{\Phi}_{t,N}^{\trans} \big]$ almost surely and thus in probability, it suffices to show that 
\begin{align*}
\frac{1}{M}\big(\Phi^{(M)}_{t,j}(S_t^{(i)})\big)_{1\leq i \leq M}^\trans Z_t^{(M)} 
\to \E_0[\Phi_{t,j}(S_t)Z_{t}]
\end{align*}
in probability for each $j =1, \dots N$. We first note that, letting $\epsilon_M^{(i)}:= Z_{M,t}^{(i)} - z_t(S_t^{(i)})$
\begin{align*}
&\Big|\frac{1}{M}\sum_{i=1}^M \Phi_{t,j}(S_t^{(i)})Z_{M,t}^{(i)} - \E_0[\Phi_{t,j}(S_t)Z_t] \Big| \\
&\quad \leq \Big|\frac{1}{M}\sum_{i=1}^M \Phi_{t,j}(S_t^{(i)})z_t(S_t^{(i)}) - \E_0[\Phi_{t,j}(S_t)Z_t]\Big| + \Big|\frac{1}{M}\sum_{i=1}^M\Phi_{t,j}(S_t^{(i)})\epsilon_M^{(i)} \Big|
\end{align*}
The first summand goes to zero in probability by the law of large numbers. Thus, we investigate the second summand using H\"older's inequality: 
\begin{align*}
& \Big|\frac{1}{M}\sum_{i=1}^M\Phi_{t,j}(S_t^{(i)})\epsilon_M^{(i)} \Big| \leq \Big( \frac{1}{M}\sum_{i=1}^M (\Phi_{t,j}(S_t^{(i)}))^2 \Big)^{1/2}\Big( \frac{1}{M}\sum_{i=1}^M (\epsilon_M^{(i)})^2 \Big)^{1/2}
\end{align*}
We see that, again by the law of large numbers, the first factor converges to $\E[ (\Phi_{t,j}(S_t))^2]$ in probability. Now we look at the second factor. By our independence assumption, $\epsilon_M^{(i)} \overset{d}{=}Z_{t,M}-Z_t$ and thus
\begin{align*}
&\Var\Big(  \Big( \frac{1}{M}\sum_{i=1}^M  (\epsilon_M^{(i)})^2  \Big)^{1/2}\Big| \calF_0\Big) \leq \E_0\Big[ \frac{1}{M}\sum_{i=1}^M  (\epsilon_M^{(i)})^2  \Big] = ||Z_{t,M}-Z_t||^2_{2},
\end{align*}
which, by assumption, goes to $0$ in probability, hence the expression goes to zero in probability. This concludes the proof.
\end{proof}

\begin{proof}[Proof of Theorem \ref{thm:simulation_convergence}]
We pove the statement by backwards induction, starting from time $t=T$. As before, the induction base follows immediately from our assumptions. Now assume $||\widehat{V}_{N,t+1}(L) -\widehat{V}_{N,t+1}^{(M)}(L)||_{2} \to 0$ in probability, as $M \to \infty$. By $L^2$-continuity we get that $||\phi_t(L_{t+1}+\widehat{V}_{N,t+1}(L)) -\phi_t(L_{t+1}+\widehat{V}_{N,t+1}^{(M)}(L))||_{2} \to 0$ in probability. But then by Lemma \ref{lem:simulation_stability_property} we immediately get that $||\widehat{V}_{N,t}(L) -\widehat{V}_{N,t}^{(M)}(L)||_{2}\to 0$ in probability.
\end{proof}

\begin{proof}[Proof of Lemma \ref{lem:Lipschitz_implies_L2}]
Note that
\begin{align*}
|| \phi_t(X)-\phi_t(Y)||^2_2 &= \E_0[| \phi_t(X)-\phi_t(Y)|^2]\leq \E_0[K^2\E_t[|X-Y|]^2] \\
&\leq K^2\E_0[\E_t[|X-Y|^2]] = K^2\E_0[|X-Y|^2] \\
&= K^2||X-Y||^2_2.
\end{align*}
Here we have used Jensen's inequality and the tower property of the conditional epectation at the second inequality and the following equality, respectively. From this $L^2$-continuity immediately follows.
\end{proof}

\begin{proof}[Proof of Lemma \ref{lem:noname1}]
By Lemma \ref{lem:R_monotonicity}, to construct upper and lower bounds for a quantity given by $\phi_{t,\alpha}(\xi)$ we may find upper and lower bounds for $\rho_t(-\xi)$ and insert them into the expression for $\phi_{t,\alpha}(\xi)$. Now
Take $X, Y \in L^p(\calF_{t+1})$ and let $Z := Y-X$. By monotnicity we get that
\begin{align*}
&\phi_{t}(X-|Z|) \leq \phi_{t}(Y) \leq\phi_{t}(X+|Z|).
\end{align*}
We now observe that
\begin{align*}
&\rho_t(-(X+|Z|)) \leq \rho_t(-X)+K\E_t[|Z|] \\
&\rho_t(-(X-|Z|)) \leq \rho_t(-X)-K\E_t[|Z|]
\end{align*}
We use this to also observe that, by the subadditivity of the $()_+$-operation,
\begin{align*}
&-\E_t[(\rho_t(-X)+K\E_t[|Z|]-X-|Z|)_+]  \\
&\quad\leq-\E_t[\rho_t(-X)-X)_+]+\E_t[(\rho_t(-|Z|) -|Z|)_+] \\
&\quad\leq-\E_t[\rho_t(-X)-X)_+]+\rho_t(-|Z|)\\
&\quad\leq-\E_t[\rho_t(-X)-X)_+]+K\E_t[|Z|].
\end{align*}
Similarly, we have that 
\begin{align*}
&-\E_t[(\rho_t(-X)-K\E_t[|Z|]-X+|Z|)_+]  \\
&\quad\leq-\E_t[\rho_t(-X)-X)_+]-\E_t[(\rho_t(-|Z|) -|Z|)_+] \\
&\quad\leq-\E_t[\rho_t(-X)-X)_+]-\rho_t(-|Z|)\\
&\quad\leq-\E_t[\rho_t(-X)-X)_+]-K\E_t[|Z|].
\end{align*}
From this we get that
\begin{align*}
&\phi_{t}(Y) \leq\phi_{t}(X+|Z|) \leq \phi_{t}(X) + K\E_t[|Z|] + \frac{1}{1+\eta}K\E_t[|Z|] \\ 
&\phi_{t}(Y) \geq\phi_{t}(X-|Z|) \geq \phi_{t}(X) - K\E_t[|Z|] - \frac{1}{1+\eta}K\E_t[|Z|] 
\end{align*}
From which Lipschitz continuity with resepct to the constant $2K$ immediately follows.
\end{proof}

\begin{proof}[Proof of Lemma \ref{lem:spectral_lipschitz}]
By subadditivity we have that 
\begin{align*}
&\rho_{t,M}(Y) -\rho_{t,M}(X) \leq \rho_{t,M}(Y-X) \leq \rho_{t,M}(-|Y-X)|),\\
&\rho_{t,M}(X) -\rho_{t,M}(Y) \leq \rho_{t,M}(X-Y) \leq \rho_{t,M}(-|Y-X)|),
\end{align*}
Now we simply note that
\begin{align*}
\rho_{t,M}(-|Y-X)|) &= -\int_0^1 F^{-1}_{t,|Y-X)|}(u) m(u)\cald u \\
&\leq  -m(0)\int_0^1 F^{-1}_{t,|Y-X)|}(u) \cald u\\
&=m(0) \E_t[|Y-X|]
\end{align*}
This concludes the proof.
\end{proof}

\begin{proof}[Proof of Lemma \ref{lem:VaR_subadditivity}]
We begin by showing \eqref{eq:VaR_ineq_1}. Let $E = \{ Z \leq \VaR_{1-\delta}(-Z)\}$. Then:
\begin{align*}
\P_t(X+Z \leq y) &\geq \P_t(E \cap \{X+Z \leq y\}) \\
&\geq  \P_t(E \cap \{X+\VaR_{1-\delta}(-Z) \leq y\})\\
&\geq \P_t(X+\VaR_{1-\delta}(-Z) \leq y) -\P(E^\complement) \\
&\geq \P_t(X\leq y-\VaR_{1-\delta}(Z)) - \delta
\end{align*}
Putting $y = \VaR_{\alpha+\delta}(-X)+\VaR_{1-\delta}(-Z) $ yields
\begin{align*}
&\P_t(X\leq y-\VaR_{1-\delta}(-Z)) - \delta \geq \alpha + \delta  - \delta = \alpha
\end{align*}
Hence $\VaR_{t,\alpha}(-(X+Z)) \leq \VaR_{\alpha+\delta}(-X)+\VaR_{1-\delta}(-Z) $. We now prove \eqref{eq:VaR_ineq_2} by applying \eqref{eq:VaR_ineq_1}
\begin{align*}
\VaR_{t,\alpha-\delta}(-X) &= \VaR_{t,\alpha-\delta}(-(X+Z + (-Z)) \\
&\leq \VaR_{t,\alpha}(-(X+Z)) + \VaR_{1-\delta}(Z),
\end{align*}
from which we get \eqref{eq:VaR_ineq_2}
\end{proof}

\begin{proof}[Proof of Corollary \ref{cor:VaR_Lipschitz}]
Let $Z=Y-X$. Now we simply note that, for any $\delta$
\begin{align*}
\VaR_{t,\alpha}(-(X+Z)) &\leq \VaR_{t,\alpha}(-(X+|Z|)) \\
&\leq \VaR_{t,\alpha+\delta}(-X) + \VaR_{t,1-\delta}(-|Z|) 
\end{align*}
By Markov's inequality, we may bound the latter summand:
\begin{align*}
&\VaR_{t,1-\delta}(-|Z|) \leq \frac{1}{\delta}\E_t[|Z|]
\end{align*}
Now for the lower bound, we similarly note
\begin{align*}
\VaR_{t,\alpha}(-(X+Z)) &\geq \VaR_{t,\alpha}(-(X-|Z|)) \\
&\geq \VaR_{t,\alpha-\delta}(-X) + \VaR_{t,1-\delta}(|Z|) 
\end{align*}
where again we may bound the second summand using Markov's inequality:
\begin{align*}
&\VaR_{t,1-\delta}(|Z|) \geq -\frac{1}{1-\delta}\E_t[|Z|]\geq -\frac{1}{\delta}\E_t[|Z|],
\end{align*}
since we have assumed $\delta<1/2$. This immediately yields that, almost surely,
\begin{align*}
\VaR_{t,\alpha}(-Y) \in \Big[ \VaR_{t,\alpha-\delta}(-X)-\frac{1}{\delta}\E_t[|X-Y|], \VaR_{t,\alpha+\delta}(-X) +\frac{1}{\delta}\E_t[|X-Y|]\Big].
\end{align*}
This immediately yields our desired result.
\end{proof}

\begin{lemma}\label{lem:R_monotonicity}
For any $X \in L^1(\calF_{t+1})$ and $R_1,R_2 \in L^0(\calF_t)$ with $R_1 \leq R_2$ a.s., 
\begin{align*}
&R_1 - \frac{1}{1+\eta}\E_t[(R_1-X)_+] \leq R_2 - \frac{1}{1+\eta}\E_t[(R_2-X)_+] \quad a.s.
\end{align*}
\end{lemma}

\begin{proof}[Proof of Lemma \ref{lem:R_monotonicity}] 
Let $R_1 \leq R_2$ a.s. and let $A_1 = \{R_1-X \geq 0\}$ and $A_2 = \{R_2-X \geq 0\}$. Note that $A_1 \subseteq A_2$ almost surely, i.e. $\P_t(A_1 \setminus A_2) = 0$ a.s. We now note that:
\begin{align*}
&R_1 - \frac{1}{1+\eta}\E_t[(R_1-X)_+]  = \big(1-\frac{1}{1+\eta}\P_t(A_1)\big)R_1 + \frac{1}{1+\eta}\E_t[\calI_{A_1}X] \\
&R_2 - \frac{1}{1+\eta}\E_t[(R_2-X)_+]  = \big(1-\frac{1}{1+\eta}\P_t(A_2)\big)R_2 +\frac{1}{1+\eta}\E_t[\calI_{A_2}X]
\end{align*}
We look at the expectation in the  first expression:
\begin{align*}
&\E_t[\calI_{A_1}X] =\E_t[\calI_{A_2}X] - \E_t[X\calI_{A_2 \setminus A_1}] \leq \E_t[\calI_{A_2}X] -\P_t(A_2\setminus A_1)R_1
\end{align*}
We now see that
\begin{align*}
&\big(1-\frac{1}{1+\eta}\P_t(A_1)\big)R_1 +\frac{1}{1+\eta}\E_t[\calI_{A_1}X] \\
&\quad\leq \big(1-\frac{1}{1+\eta}\P_t(A_1)\big)R_1  + \frac{1}{1+\eta}\E_t[\calI_{A_2}X]  - \frac{1}{1+\eta}\P_t(A_2\setminus A_1)R_1 \\
&\quad= \big(1-\frac{1}{1+\eta}\P_t(A_2)\big)R_1+ \frac{1}{1+\eta}\E_t[\calI_{A_2}X] \\
&\quad\leq \big(1-\frac{1}{1+\eta}\P_t(A_2)\big)R_2+ \frac{1}{1+\eta}\E_t[\calI_{A_2}X]
\end{align*}
This concludes the proof.
\end{proof}

\begin{proof}[Proof of Theorem \ref{thm:phi_Lipschitz}]
Let $Z=Y-X$. Note that
\begin{align*}
&\phi_{t,\alpha}(Y) =\phi_{t,\alpha}(X+Z) \leq \phi_{t,\alpha}(X+|Z|).
\end{align*}
As for the $\VaR$-part of $\phi_{t,\alpha}$, including that in the expectation, we note that by Lemma \ref{lem:VaR_subadditivity} $\VaR_{t,\alpha}(-(X+|Z|))\leq \VaR_{t,\alpha+\delta}(-X)+\VaR_{t,1-\delta}(-|Z|)$. We now note that, by subadditivity of $x\mapsto (x)_+$,
\begin{align*}
&-\E_t[(\VaR_{t,\alpha+\delta}(-X)+\VaR_{t,1-\delta}(-|Z|) -X-|Z|)_+]  \\
&\quad\leq-\E_t[(\VaR_{t,\alpha+\delta}(-X)-X)_+]+\E_t[(\VaR_{t,1-\delta}(-|Z|) -|Z|)_+] \\
&\quad\leq -\E_t[(\VaR_{t,\alpha+\delta}(-X)-X)_+] + \VaR_{t,1-\delta}(-|Z|). 
\end{align*}
Hence
\begin{align*}
\phi_{t,\alpha}(X+|Z|) &\leq \phi_{t,\alpha+\delta}(X) + \frac{2+\eta}{1+\eta}\VaR_{t,1-\delta}(-|Z|) \\
&\leq \phi_{t,\alpha+\delta}(X) + \frac{2}{\delta}\E_t[|Z|].
\end{align*}
Here we have used the Markov's inequality bound from Corollary \ref{cor:VaR_Lipschitz}. 

We now similarly construct a lower bound for $\phi_{t,\alpha}(Y)$:
\begin{align*}
&\phi_{t,\alpha}(Y) =\phi_{t,\alpha}(X+Z) \geq \phi_{t,\alpha}(X-|Z|) 
\end{align*}
Again, for the $\VaR$-part, we note that by Lemma \ref{lem:VaR_subadditivity} $\VaR_{t,\alpha}(-(X-|Z|))\geq \VaR_{t,\alpha-\delta}(-X)+\VaR_{t,1-\delta}(|Z|)$. We now analyze the resulting expected value part, using subadditivity of $x\mapsto (x)_+$:
\begin{align*}
&-\E_t[(\VaR_{t,\alpha+\delta}(-X)+\VaR_{t,1-\delta}(|Z|) -X+|Z|)_+]  \\
&\quad\geq -\E_t[(\VaR_{t,\alpha+\delta}(-X)-X)_+] - \E_t[(\VaR_{t,\alpha+\delta}(|Z|)+|Z|)_+] \\
&\quad\geq  -\E_t[(\VaR_{t,\alpha+\delta}(-X)-X)_+] -\E_t[|Z|]
\end{align*}
Hence we get the lower bound
\begin{align*}
\phi_{t,\alpha}(X-|Z|) &\leq \phi_{t,\alpha-\delta}(X) +\VaR_{t,1-\delta}(|Z|)+\frac{1}{1+\eta}\E_t[|Z|] \\
&\geq \phi_{t,\alpha-\delta}(X) -\frac{2}{\delta}\E_t[|Z|].
\end{align*}
Here, again, we have used the Markov's inequality bound from Corollary \ref{cor:VaR_Lipschitz}. Hence we have shown that
\begin{align*}
&\phi_{t,\alpha}(Y) \in \Big[\phi_{t,\alpha-\delta}(X) -\frac{2}{\delta}\E_t[|X-Y|] ,\phi_{t,\alpha+\delta}(X)+\frac{2}{\delta}\E_t[|X-Y|] \Big],
\end{align*}
from which \eqref{eq:phi_Lipschitz} immediately follows.
\end{proof}
\begin{proof}[Proof of Corollary \ref{cor:phi_L2cont}]
Choose a sequence $\delta_n \to 0$ such that $\frac{2}{\delta_n^2}\E[|X-X_n|^2] \to 0$ with $\alpha +\delta_n<1$ and $\delta_n<1/2$. We now use the following inequality, which follows from the monotonicity of $\phi_{t,\alpha}$ in $\alpha$:
\begin{align*}
&|\phi_{t,\alpha}(X) - \phi_{t,\alpha}(X_n)| \\
&\quad \leq \phi_{t,\alpha+\delta_n}(X)-\phi_{t,\alpha-\delta_n}(X) + \inf_{|\epsilon|<\delta_n}|\phi_{t,\alpha+\epsilon}(X)-\phi_{t,\alpha}(X_n)| \\
&\quad \leq \phi_{t,\alpha+\delta_n}(X)-\phi_{t,\alpha-\delta_n}(X) + \frac{2}{\delta_n}\E[|X-X_n|\mid \calH_t]
\end{align*}
By $L^2$-convergence and our choice of $\delta_n$, the last term clearly goes to $0$ by our assumptions.

As for the first summand, we see that for any sequence $\delta_n \to 0$, $\phi_{t,\alpha-\delta_n}(X)-\phi_{t,\alpha+\delta_n}(X) \to 0$ almost surely (by the continuity assumption of $\VaR_{t,u}$ at $\alpha$) and furthermore it is a decreasing sequence of nonnegative random variables in $L^2(\calH_t)$. Hence by Lebesgue's monotone convergence theorem $||\phi_{t,\alpha-\delta_n}(X)-\phi_{t,\alpha+\delta_n}(X)||_ 2 \to 0$. This concludes the proof.
\end{proof}
\begin{proof}[Proof of Lemma \ref{lem:noname2}]
By Lemma \ref{cor:phi_L2cont}, $\phi_{t,\alpha}$ is $L^2$ continuous with respect to limit objects with a.s. continuous $t$-conditional distributions. Hence the proof is completely analogous to that of Lemma 1.
\end{proof}

\begin{proof}[Proof of Lemma \ref{lem:VaR_L2_continuity}]
Fix $\epsilon >0$. we want to show that $\P(||\phi_{t,\alpha}(\beta^\trans \mathbf{\Phi}_{t+1,N})-\phi_{t,\alpha}(\beta_n^\trans \mathbf{\Phi}_{t+1,N})||_ 2>\epsilon) \to 0$ as $n \to \infty$. We first note that, for any $\delta \in (0,1-\alpha)$  with $\delta< 1/2$, we have an inequality similar to that in the proof of Corollary \ref{cor:phi_L2cont}:
\begin{align*}
&||\phi_{t,\alpha}(\beta^\trans \mathbf{\Phi}_{t+1,N})-\phi_{t,\alpha}(\beta_n^\trans \mathbf{\Phi}_{t+1,N})||_ 2  \\
&\quad\leq ||\phi_{t,\alpha-\delta }(\beta^\trans \mathbf{\Phi}_{t+1,N})-\phi_{t,\alpha+\delta}(\beta^\trans \mathbf{\Phi}_{t+1,N})||_2 \\
&\quad+\inf_{|\xi|\leq \delta}||\phi_{t,\alpha+\xi}(\beta^\trans \mathbf{\Phi}_{t+1,N})-\phi_{t,\alpha}(\beta_n^\trans \mathbf{\Phi}_{t+1,N})||_ 2
\end{align*}
If we look at the first summand, we see that for any sequence $\delta_n \to 0$, $\phi_{t,\alpha-\delta_n }(\beta^\trans \mathbf{\Phi}_{t+1,N})-\phi_{t,\alpha+\delta_n}(\beta^\trans \mathbf{\Phi}_{t+1,N}) \to 0$ almost surely (by a.s. continuity) and furthermore it is a decreasing sequence of nonnegative random variables. Hence by Lebesgue's monotone convergence theorem $||\phi_{t,\alpha-\delta_n }(\beta^\trans \mathbf{\Phi}_{t+1,N})-\phi_{t,\alpha+\delta_n}(\beta^\trans \mathbf{\Phi}_{t+1,N})||_ 2 \to 0$ as a sequence of constants, since the expression is independent of $D$. 

We now apply Theorem \ref{thm:phi_Lipschitz} to the second term to see that
\begin{align*}
\inf_{|\xi|\leq \delta}||\phi_{t,\alpha+\xi}(\beta^\trans \mathbf{\Phi}_{t+1,N})-\phi_{t,\alpha}(\beta_n^\trans \mathbf{\Phi}_{t+1,N})||_ 2 &\leq \frac{2}{\delta}||(\beta-\beta_n)^\trans\mathbf{\Phi}_{t+1,N}||_2\\
&\leq  \frac{2}{\delta}||\beta-\beta_n||_\infty ||\mathbf{\Phi}_{t+1,N}||_2
\end{align*}
Note that $||\mathbf{\Phi}_{t+1,N}||_2 :=K_\Phi$ is just a constant. We now note that, as $||\beta-\beta_n||_\infty \to 0$ in probability, then for any fixed  $\epsilon>0$, it is possible to choose a sequence $\delta_n\to 0$ such that $\P(\frac{2K_\Phi}{\delta_n}||\beta-\beta_n||_\infty>\epsilon) \to 0$. Hence, for any fixed $\epsilon >0$, $\P(||\phi_{t,\alpha}(\beta^\trans \mathbf{\Phi}_{t+1,N})-\phi_{t,\alpha}(\beta_n^\trans \mathbf{\Phi}_{t+1,N})||_ 2>\epsilon) \to 0$ as $n \to \infty$.
\end{proof}

\begin{proof}[Proof of Theorem \ref{thm:noname3}]
We prove the statement by backwards induction, starting from time $t=T$. The induction base is trivial. Now assume that the statement holds for time $t+1$. But then, by Lemma \ref{lem:VaR_L2_continuity}, $||\phi_{t,\alpha}(L_{t+1}+\widehat{V}_{N,t,\alpha}(L))- \phi_{t,\alpha}(L_{t+1}+\widehat{V}^{(M)}_{N,t,\alpha}(L))||_2\to 0$ in probability. Hence, we get immediately by Lemma \ref{lem:simulation_stability_property} that $||\widehat{V}_{N,t,\alpha}(L) -\widehat{V}_{N,t,\alpha}^{(M)}(L)||_{2} \to 0$ in probability. This concludes the proof.
\end{proof}


\newpage

\begin{thebibliography}{99}
\bibitem{Artzner-Delbaen-Eber-Heath-Ku-07}
Philippe Artzner, Freddy Delbaen, Jean-Marc Eber, David Heath and Hyejin Ku (2007),
Coherent multiperiod risk adjusted values and Bellman's principle.
\emph{Annals of Operations Research}, 152, 5-22.

\bibitem{Barrera-18}
D. Barrera, S. Cr\'epey, B. Diallo, G. Fort, E. Gobet and U. Stazhynski (2018), Stochastic approximation schemes for economic capital and risk margin computations. \emph{HAL \textless hal-01710394\textgreater}

\bibitem{Brodie-15}
Mark Broadie, Yiping Du and Ciamac C. Moallemi (2015),
Risk estimation via regression.
\emph{Operations Research}, 63 (5) 1077-1097.

\bibitem{Cheridito-Delbaen-Kupper-06}
Patrick Cheridito, Freddy Delbaen and Michael Kupper (2006),
Dynamic monetary risk measures for bounded discrete-time processes.
\emph{Electronic Journal of Probability}, 11, 57-106.

\bibitem{Cheridito-Kupper-11}
Patrick Cheridito and Michael Kupper (2011),
Composition of time-consistent dynamic monetary risk measures in discrete time.
\emph{International Journal of Theoretical and Applied Finance}, 14 (1), 137-162.

\bibitem{Cheridito-Kupper-09}
Patrick Cheridito and Michael Kupper (2009),
Recursiveness of indifference prices and translation-invariant preferences.
\emph{Mathematics and Financial Economics}, 2 (3), 173-188.

\bibitem{Clement-etal-02}
Emanuelle Cl\'ement, Damien Lamberton, and Philip Protter(2002),
 An analysis of a least squares regression method for american option pricing.
\emph{ Finance and Stochastics}, 6, 449–471.

\bibitem{Commission-del-reg-15}
European Commission (2015),
Commission delegated regulation (EU) 2015/35 of 10 October 2014.
\emph{Official Journal of the European Union}.

\bibitem{Delong-19}
\L{}ukasz Delong, Jan Dhaene and Karim Barigou (2019), 
Fair valuation of insurance liability cash-flow streams in continuous
time: Applications. 
\emph{ASTIN Bulletin}, 49 (2), 299-333.

\bibitem{Detlefsen-Scandolo-05}
Kai Detlefsen and Giacomo Scandolo (2005), 
Conditional and dynamic convex risk measures. 
\emph{Finance and Stochastics}, 9, 539-561.

\bibitem{Dhaene-17}
Jan Dhaene, Ben Stassen, Karim Barigou, Daniël Linders and Ze Chen (2017)
Fair valuation of insurance liabilities: Merging actuarial judgement
and Market-Consistency. 
\emph{Insurance: Mathematics and Economics}, 76, 14-27.

\bibitem{Engsner-Lindholm-Lindskog-17}
Hampus Engsner, Mathias Lindholm and Filip Lindskog (2017),
Insurance valuation: A computable multi-period cost-of-capital approach.
\emph{Insurance: Mathematics and Economics}, 72, 250-264. 

\bibitem{Engsner-Lindskog-20}
Hampus Engsner and Filip Lindskog (2020),
Continuous-time limits of multi-period cost-of-capital margins.
\emph{Statistics and Risk Modelling}, forthcoming  (DOI: https://doi.org/10.1515/strm-2019-0008)

\bibitem{Egloff-05}
Daniel Egloff (2005), Monte Carlo algorithms for optimal stopping and statistical learning, 
\emph{The Annals of Applied Probability}, 15 (2), 1396-1432.

\bibitem{Foellmer-Schied-16}
Hans F\"ollmer and Alexander Schied (2016),
\emph{Stochastic finance: An introduction in discrete time}, 4th edition,
De Gruyter Graduate.

\bibitem{Longstaff-Schwartz-01}
Francis A. Longstaff and Eduardo S. Schwartz (2001). Valuing American options by simulation:
A simple least-squares approach. 
\emph{The Review of Financial Studies}, 14 (1),113-147.

\bibitem{Moehr-11}
Christoph M\"ohr (2011),
Market-consistent valuation of insurance liabilities by cost of capital.
\emph{ASTIN Bulletin}, 41, 315-341.


\bibitem{Pelsser-G-20} Antoon Pelsser and Ahmad Salahnejhad Ghalehjooghi (2020). Time-consistent and market-consistent actuarial valuation of the participating pension contract. 
\emph{Scandinavian Actuarial Journal}, forthcoming (DOI: https://doi.org/10.1080/03461238.2020.1832911)

\bibitem{Pelsser-G-16} 
Antoon Pelsser and Ahmad Salahnejhad Ghalehjooghi (2016). Time-consistent actuarial valuations. 
\emph{Insurance: Mathematics and Economics}, 66, 97-112.

\bibitem{Pelsser-Stadje-14} 
Antoon Pelsser and Mitja Stadje (2014), 
Time-consistent and market-consistent evaluations.
\emph{Mathematical Finance}, 24 (1), 25-65.

\bibitem{Stentoft-04} Lars Stentoft (2004), Convergence of the least squares Monte Carlo approach to American option valuation. 
\emph{Management Science}, 50 (9), 1193-1203. 

\bibitem{Tsitsiklis-vanRoy-01}
John N. Tsitsiklis and Benjamin Van Roy (2001). Regression methods for pricing complex
American-style options. 
\emph{IEEE Transactions On Neural Networks}, 12 (4), 694-703.

\bibitem{Zanger-09} Daniel Z. Zanger (2009), Convergence of a least-squares Monte Carlo
algorithm for bounded approximating sets. 
\emph{Applied Mathematical Finance}, 16 (2),  123-150.

\bibitem{Zanger-13} Daniel Z. Zanger (2013), Quantitative error estimates for a least-squares Monte
Carlo algorithm for American option pricing. 
\emph{Finance and Stochastics}, 17, 503-534.

\bibitem{Zanger-18} Daniel Z. Zanger (2018), Convergence of a least-squares Monte Carlo algorithm for American option pricing with dependent sample data. 
\emph{Mathematical Finance}, 28 (1), 447-479.
\end{thebibliography}
\end{document}